\documentclass[a4paper,UKenglish,numberwithinsect]{lipics}

\usepackage{amsmath,amssymb,amstext}
\usepackage{proof}
\usepackage{wrapfig}
\usepackage{framed}
\usepackage{tikz}
\usetikzlibrary{arrows,automata,backgrounds,calc,chains,fadings,folding,decorations.fractals,fit,patterns,positioning,mindmap,shadows,shapes.geometric,shapes.symbols,through,trees,decorations.markings,decorations.text}
\colorlet{dblue}{blue!40!black}
\usepackage{enumitem}

\newcommand{\msf}{\mathsf}
\renewcommand{\mit}{\mathit}

\newcommand{\ap}{{\scriptsize@}}

\theoremstyle{definition}
\newtheorem{notation}[theorem]{Notation}
\newtheorem*{mremark}{Remark}

\renewcommand{\theenumi}{\roman{enumi}}

\newcommand{\pairlft}{{\langle}}
\newcommand{\pairrgt}{{\rangle}}
\newcommand{\pairsep}{{,\,}}
\newcommand{\pairstr}[1]{\pairlft#1\pairrgt}
\newcommand{\pair}[2]{\pairstr{#1\pairsep#2}}

\newcommand{\sBNFis}{{{:}{:}{=}}}

\newcommand{\BNFor}{\mathrel{|}}
\newcommand{\coBNFis}{\mathrel{\sBNFis^{\text{co}}}}

\newcommand{\atrs}{\mathcal{R}}

\newcommand{\sred}{{\rightarrow}}

\newcommand{\rerat}[2]{\mathrel{\sred_{#1,#2}}}

\newcommand{\mred}{\to\hspace{-2.8mm}\to}

\newcommand{\ieq}{\stackrel{\infty}{=}}

\newcommand{\harpoondown}{\mathrel{\ooalign{$\leftharpoondown$\cr$\rightharpoondown$\cr}}}

\newcommand{\ieqdown}{\stackrel{\infty}{{\harpoondown}}}%
\newcommand{\ired}{\to^\infty}%
\newcommand{\iredi}{\mathrel{\reflectbox{$\ired$}}}
\newcommand{\ireddown}{\rightharpoondown^\infty}%
\newcommand{\ireddownfin}{\stackrel{\makebox(0,0){\raisebox{3pt}{{\tiny$<$}}}}{\rightharpoondown}^\infty}%
\newcommand{\ireddownmfin}{\stackrel{\makebox(0,0){\raisebox{3pt}{{\tiny($<$)}}}}{\rightharpoondown}^\infty}%
\newcommand{\rstep}{\to_{\varepsilon}}

\newcommand{\ibi}{\stackrel{\infty}{\to}}%
\newcommand{\ibiinv}{\stackrel{\infty}{\leftarrow}}%
\newcommand{\ibidown}{\stackrel{\infty}{\rightharpoondown}}%

\newcommand{\redord}{\to_{\textit{ord}}}
\newcommand{\redorddown}{\rightharpoondown_{\textit{ord}}}
\newcommand{\iredord}{\ired_{\textit{ord}}}

\newcommand{\iredorddown}{\ireddown_{\textit{ord}}}

\newcommand{\fun}[1]{\mathsf{#1}}

\newcommand{\relcomp}{\circ}

\newcommand{\nat}{\mathbb{N}}

\newcommand{\sarity}{\mit{ar}}
\newcommand{\arity}[1]{\sarity(#1)}
\newcommand{\avars}{\mathcal{X}}
\newcommand{\vars}[1]{\mathcal{V}\hspace{-.1ex}\mit{ar}(#1)}

\newcommand{\ster}{\mathit{Ter}}
\newcommand{\ter}[2]{\ster(#1,#2)}
\newcommand{\siter}{\ster^{\infty}}
\newcommand{\iter}[2]{\siter(#1,#2)}
\newcommand{\posemp}{\varepsilon}
\newcommand{\apos}{p}

\newcommand{\pos}[1]{\mathcal{P}\!os(#1)}

\newcommand{\subtrm}[2]{#1|_{#2}}
\newcommand{\asubst}{\sigma}

\newcommand{\slfp}{\mu}
\newcommand{\lfp}[2]{\slfp {#1}.\,#2}
\newcommand{\sgfp}{\nu}
\newcommand{\gfp}[2]{\sgfp {#1}.\,#2}

\newcommand{\down}[1]{\overline{#1}}

\newcommand{\id}{\mathrm{Id}}%

\newcommand{\tlat}{L}

\newcommand{\Pow}{\mathcal{P}}

\newcommand{\rsplit}{\ensuremath{\msf{split}}}
\newcommand{\rlift}{\ensuremath{\msf{lift}}}
\newcommand{\rid}{\ensuremath{\msf{id}}}%

\newcommand{\nest}{\textit{der}}
\newcommand{\fp}{\textit{fp}}

\newcommand{\boxx}[1]{\colorbox[rgb]{0.99,0.78,0.07}{\kern0.15em#1\kern0.15em}\quad}

\newcommand{\hole}{\raisebox{-2pt}{\scalebox{.7}[1.5]{$\Box$}}}

\newcommand{\sdefd}{:=}
\newcommand{\defd}{\mathrel{\sdefd}}

\newcommand{\aes}{\atrs}%

\renewcommand{\emptyset}{\varnothing}

\newcommand{\fap}[2]{#1(#2)}
\newcommand{\bfap}[3]{\fap{#1}{#2,#3}}

\newcommand{\smetric}{\mathrm{d}}
\newcommand{\metric}{\bfap{\smetric}}

\newcommand{\smktree}{\mathfrak{T}}
\newcommand{\mktree}{\fap{\smktree}}
\newcommand{\smktreedown}{\smktree'}
\newcommand{\mktreedown}{\fap{\smktreedown}}

\newcommand{\smktreedownfin}{\smktree'_{<}}
\newcommand{\mktreedownfin}{\fap{\smktreedownfin}}
\newcommand{\smktreedownmfin}{\smktree'_{(<)}}
\newcommand{\mktreedownmfin}{\fap{\smktreedownmfin}}

\newcommand{\sto}{\rightsquigarrow}

\newcommand{\tail}{\msf{f}}

\title{A Coinductive Framework for Infinitary Rewriting and Equational Reasoning (Extended Version\footnote{%
  This current paper is an extended version of~\cite{endr:hans:hend:polo:silv:2013}.
  We have included a detailed comparison of 
  our notion of infinitary equational reasoning $\ieq$
  with the notion $S_{E}(\to)$ from~\cite{kahr:2013}, see Section~\ref{sec:venn}.
})}

\author[1]{J\"{o}rg Endrullis}
\author[2]{Helle Hvid Hansen}
\author[1]{Dimitri Hendriks} 
\author[3]{Andrew Polonsky}
\author[4]{Alexandra Silva}
\authorrunning{Endrullis, Hansen, Hendriks, Polonsky, and Silva}
\affil[1]{
  Department of Computer Science, VU University Amsterdam, The Netherlands,
  \texttt{\{j.endrullis | r.d.a.hendriks\}@vu.nl}
}

\affil[2]{
  Department of Engineering Systems and Services,
  Delft University of Technology, The Netherlands,
  \texttt{h.h.hansen@tudelft.nl}
}

\affil[3]{
  Institut Galil\'{e}e, Universit\'{e} Paris 13, France,
  \texttt{andrew.polonsky@gmail.com}
}

\affil[4]{
  Department of Computer Science, Radboud University Nijmegen, The Netherlands,
  \texttt{alexandra@cs.ru.nl}
}


\volumeinfo{Maribel Fern\'andez}
{1}
{26th International Conference on Rewriting Techniques and Applications (RTA'15)}
{36}
{1}
{1}%
\EventShortName{RTA 2015}%
\DOI{10.4230/LIPIcs.RTA.2015.x}
\serieslogo{}

\begin{document}

\maketitle

\begin{abstract}
  We present a coinductive framework for defining infinitary analogues of equational reasoning 
  and rewriting in a uniform way. We define the relation $\ieq$, a notion of 
  infinitary equational reasoning, and $\ired$, the standard notion of infinitary rewriting as follows:
  \begin{align*}
    {\ieq} &\;\;\defd\;\; \gfp{R}{(=_{\aes} \cup \mathrel{\down{R}})^*} \\ %
    {\ired} &\;\;\defd\;\; \lfp{R}{\gfp{S}{(\to_{\atrs} \cup \mathrel{\down{R}})^*\relcomp \down{S}}} %
  \end{align*}
  where $\slfp$ and $\sgfp$ are the least and greatest fixed-point operators, respectively,
  and where 
  \begin{align*}
    \down{R} \;\defd\; \{\,\pair{f(s_1,\ldots,s_n)}{\,f(t_1,\ldots,t_n)} \mid f \in \Sigma,\, s_1 \mathrel{R} t_1,\ldots,s_n \mathrel{R} t_n\,\} \,\cup\, \id  \,.
  \end{align*}
  
  The setup captures rewrite sequences of arbitrary ordinal length, 
  but it has neither the need for ordinals nor for metric convergence. 
  This makes the framework especially suitable for formalizations in theorem provers.
\end{abstract}

\section{Introduction}\label{sec:intro}

We present a coinductive framework for defining infinitary equational reasoning and infinitary rewriting
in a uniform way. The framework is free of ordinals, metric convergence and partial orders which have been
essential in earlier definitions of the 
concept of infinitary rewriting~\cite{ders:kapl:plai:1991,kenn:klop:slee:vrie:1995a,klop:vrij:2005,kenn:vrie:2003,kahr:2013,bahr:2010,bahr:2010b,bahr:2012,endr:hend:klop:2012}.

Infinitary rewriting is a generalization of the ordinary finitary rewriting
to infinite terms and infinite reductions (including reductions of ordinal length greater than $\omega$).
For the definition of rewrite sequences of ordinal length,
there is a design choice concerning the exclusion of jumps at limit ordinals,
as illustrated in the ill-formed rewrite sequence 
  {\abovedisplayskip.75ex 
   \belowdisplayskip.75ex
  \begin{align*}
    \underbrace{\fun{a} \to \fun{a} \to \fun{a} \to \cdots}_{\text{$\omega$-many steps}} \;\fun{b} \to \fun{b}
  \end{align*}}%
where the rewrite system is $\atrs = \{\,\fun{a} \to \fun{a},\, \fun{b}\to\fun{b}\,\}$.
The rewrite sequence remains for $\omega$ steps at $\fun{a}$ and in the limit step `jumps' to $\fun{b}$.
To ensure connectedness at limit ordinals, the usual choices are:
\begin{enumerate}
  \item \emph{weak convergence} (also called `Cauchy convergence'), where it suffices that the sequence of terms converges towards the limit term, and
  \item \emph{strong convergence}, which additionally requires that the `rewriting activity', i.e., the depth of the rewrite steps,
    tends to infinity when approaching the limit.
\end{enumerate}
The notion of strong convergence incorporates the flavor of `progress', or `productivity',
in the sense that there is only a finite number of rewrite steps at every depth.
Moreover, it leads to a more satisfactory metatheory where redex occurrences can be 
traced over limit steps.

While infinitary rewriting has been studied extensively,
notions of infinitary equational reasoning have not received much attention.
One of the few works in this area is~\cite{kahr:2013} by Kahrs,
see \textit{Related Work} below.
The reason is that the usual definition of infinitary rewriting is
based on ordinals to index the rewrite steps, 
and hence the rewrite direction is incorporated from the start.
This is different for the framework we propose here, 
which enables us to define several natural notions:
infinitary equational reasoning, bi-infinite rewriting, and the standard concept of infinitary rewriting.
All of these have strong convergence `built-in'.

We define \emph{infinitary equational reasoning} 
with respect to a system of equations $\aes$, as a relation~${\ieq}$
on potentially infinite terms %
by the following 
mutually coinductive rules:
\begin{gather}
  \begin{aligned}
    \infer=
    {s \ieq t}
    {s \mathrel{(=_\aes \cup \ieqdown)^*} t}
    &&\qquad\qquad\qquad&&
    \infer=
    {f(s_1,s_2,\ldots,s_n) \ieqdown f(t_1,t_2,\ldots,t_n)}
    {s_1 \ieq t_1 & \cdots & s_n \ieq t_n}
  \end{aligned}
  \label{rules:intro:ieq}
\end{gather}
The relation ${\ieqdown}$ stands for infinitary equational reasoning below the root.
The coinductive nature of the rules means that the proof trees
need not be well-founded.
Reading the rules bottom-up,
the first rule allows for an arbitrary, but finite, number of rewrite steps at any finite depth (of the term tree).
The second rule enforces that we eventually proceed with the arguments, and hence the activity tends to infinity.

\begin{example}\label{ex:Ca:a}
  Let $\aes$ consist of the equation $\fun{C}(\fun{a}) = \fun{a}$.
\end{example}
  \vspace{-1.5ex}

  \begin{wrapfigure}{r}{5cm}
    \vspace{-8ex}
    \begin{framed}
    \vspace{-1.5ex}
    \begin{align*}
      \infer=
      {\fun{C}^\omega \ieq \fun{a}}
      {
        \infer=
        {\fun{C}^\omega \ieqdown \fun{C}(\fun{a})}
        {\infer={\fun{C}^\omega \ieq \fun{a}}
            {\makebox(0,0){
              \hspace{-13mm}\begin{tikzpicture}[baseline=12ex]
              \draw [->,thick,dotted] (0,0) -- (0,1mm) to[out=90,in=80] (-11mm,-1mm) to[out=-100,in=-160,looseness=1.6] (-0mm,-13mm);
              \end{tikzpicture}
            }}
        }
        &
        \fun{C}(\fun{a}) =_\aes \fun{a}
      }
    \end{align*}
    \vspace{-6ex}
    \end{framed}
    \vspace{-4ex}
    \caption{Derivation of $\fun{C}^\omega \ieq \fun{a}$.}
    \vspace{-3ex}
    \label{fig:ieq}
  \end{wrapfigure}
  \noindent
  We write $\fun{C}^\omega$ to denote the infinite term $\fun{C}(\fun{C}(\fun{C}(\ldots)))$,
  the solution of the equation $X = \fun{C}(X)$.
  Using the rules~\eqref{rules:intro:ieq}, we can derive 
  $\fun{C}^\omega \ieq \fun{a}$ as shown in Figure~\ref{fig:ieq}.
  This is an infinite proof tree as indicated by the loop 
  \raisebox{.5ex}{\tikz \draw [->,thick,dotted] (0,0) -- (5mm,0mm);}
  in which the sequence 
  $\fun{C}^\omega \ieqdown \fun{C}(\fun{a}) =_\aes \fun{a}$
  is written 
  by juxtaposing
  $\fun{C}^\omega \ieqdown \fun{C}(\fun{a})$ and $\fun{C}(\fun{a}) =_\aes \fun{a}$.

Using the greatest fixed-point constructor~$\sgfp$, we can define $\ieq$ equivalently as follows:
\begin{align}
  {\ieq} &\;\;\defd\;\; \gfp{R}{(=_\aes \cup \mathrel{\down{R}})^*} \,,
  \label{eq:ieq}
\end{align}
where $\down{R}$, corresponding to the second rule in~\eqref{rules:intro:ieq}, is defined by
\begin{align}
  \down{R} \;\defd\; \{\,\pair{f(s_1,\ldots,s_n)}{\,f(t_1,\ldots,t_n)} \mid f \in \Sigma,\; s_1 \mathrel{R} t_1,\,\ldots,s_n \mathrel{R} t_n\,\} \,\cup\, \id  \,.
\end{align}
This is a new and interesting notion of infinitary (strongly convergent) equational reasoning.

%

Now let $\atrs$ be a term rewriting system (TRS).
If we use $\to_\atrs$ instead of $=_\aes$ in the rules~\eqref{rules:intro:ieq}, %
we obtain what we call \emph{bi-infinite rewriting $\ibi$}\,:
\begin{gather}
  \begin{aligned}
    \infer=
    {s \ibi t}
    {s \mathrel{(\to_\atrs \cup \ibidown)^*} t}
    &&\qquad\qquad\qquad&&
    \infer=
    {f(s_1,s_2,\ldots,s_n) \ibidown f(t_1,t_2,\ldots,t_n)}
    {s_1 \ibi t_1 & \cdots & s_n \ibi t_n}
  \end{aligned}
  \label{rules:intro:ibi}
\end{gather}
corresponding to the following fixed-point definition:
\begin{align}
  {\ibi} &\;\;\defd\;\; \gfp{R}{(\to_\atrs \cup \mathrel{\down{R}})^*} \,.
  \label{eq:intro:ibi}
\end{align}
We write $\ibi$ to distinguish bi-infinite rewriting 
from the standard notion $\ired$ of (strongly convergent) infinitary rewriting~\cite{tere:2003}.
The symbol $\infty$ is centered above $\to$ in $\ibi$ 
to indicate that bi-infinite rewriting is `balanced', 
in the sense that it allows rewrite sequences to be extended infinitely forwards, but also infinitely backwards.
Here backwards does \emph{not} refer to reversing the arrow $\leftarrow_{\varepsilon}$.
For example, for $\atrs = \{\, \fun{C}(\fun{a}) \to \fun{a} \,\}$
we have the backward-infinite rewrite sequence $\cdots \to \fun{C}(\fun{C}(\fun{a})) \to \fun{C}(\fun{a}) \to \fun{a}$
and hence $\fun{C}^\omega \ibi \fun{a}$.
The proof tree for $\fun{C}^\omega \ibi \fun{a}$
has the same shape as the proof tree displayed in Figure~\ref{fig:ieq};
the only difference is that $\ieq$ is replaced by $\ibi$ and $\ieqdown$ by $\ibidown$.
In contrast, the standard notion $\ired$ of infinitary rewriting only takes into account forward limits
and we do \emph{not} have $\fun{C}^\omega \ired \fun{a}$.

We have the following strict inclusions:
\begin{align*}
  {\ired} \;\;\subsetneq\;\; {\ibi} \;\;\subsetneq\;\; {\ieq} \;\,.
\end{align*}
In our framework, these inclusions follow directly from the fact 
that the proof trees for $\ired$ (see below)
are a restriction of the proof trees for $\ibi$
which in turn are a restriction of the proof trees for $\ieq$. 
It is also easy to see that each inclusion is strict.
For the first, see above. For the second, just note that $\ibi$ is not symmetric.

Finally, by a further restriction of the proof trees,
we obtain the standard concept of (strongly convergent) infinitary rewriting $\ired$. 
Using least and greatest fixed-point operators, we define:
\begin{align}
  {\ired} &\;\;\defd\;\; \lfp{R}{\gfp{S}{(\to \cup \mathrel{\down{R}})^*\relcomp \down{S}}} \,,
  \label{eq:intro:ired}
\end{align}
where $\circ$ denotes relational composition.
Here $R$ is defined inductively,
and $S$ is defined coinductively.
Thus only the last step in the sequence $(\to \cup \mathrel{\down{R}})^*\relcomp \down{S}$ is coinductive.
This corresponds to the following fact about reductions $\sigma$ of ordinal length:
every strict prefix of $\sigma$ must be shorter than $\sigma$ itself,
while strict suffixes may have the same length as $\sigma$.

If we replace $\slfp$ by $\sgfp$ in \eqref{eq:intro:ired}, 
we get a definition equivalent to~$\ibi$ defined by~\eqref{eq:intro:ibi}.
To see that it is at least as strong, note that $\id \subseteq \down{S}$.

Conversely, $\ired$ can be obtained by a restriction of the proof trees obtained by the rules~\eqref{rules:intro:ibi} for $\ibi$.
Assume that in a proof tree using the rules~\eqref{rules:intro:ibi},
we mark those occurrences of $\ibidown$ that 
are followed by another step in the premise of the rule 
(i.e., those that are not the last step in the premise).
Thus we split $\ibidown$ into $\ireddown$ and~$\ireddownfin$.
Then the restriction to obtain the relation $\ired$ is to forbid infinite nesting of marked symbols~$\ireddownfin$.
This marking is made precise in the following rules:
\begin{align}
  \begin{aligned}
  \infer=
  {s \ired t}
  {s \mathrel{(\to \cup \ireddownfin)^*} \relcomp \ireddown t}
  &\hspace{.4cm}&  
  \infer=
  {f(s_1,s_2,\ldots,s_n) \ireddownmfin f(t_1,t_2,\ldots,t_n)}
  {s_1 \ired t_1 & \cdots & s_n \ired t_n}
  &\hspace{.4cm}&
  \infer=
  {s \ireddownmfin s}
  {}
  \end{aligned}
  \label{rules:restrict}
\end{align}
Here $\ireddown$ stands for infinitary rewriting below the root,
and $\ireddownfin$ is its marked version.
The symbol $\ireddownmfin$ stands for both $\ireddown$ and $\ireddownfin$.
Correspondingly, the rule in the middle is an abbreviation for two rules.
The axiom ${s \ireddown s}$ serves to `restore' reflexivity, that is,
it models the identity steps in $\down{S}$ in~\eqref{eq:intro:ired}.
Intuitively, $s \ireddownfin t$ can be thought of as an infinitary rewrite sequence 
below the root, shorter than the sequence we are defining.

We have an infinitary strongly convergent rewrite sequence from $s$ to $t$ 
if and only if $s \ired t$ can be derived by the rules~\eqref{rules:restrict}
in a (not necessarily well-founded) proof tree without infinite nesting of $\ireddownfin$,
that is, proof trees in which all paths (ascending through the proof tree) contain only
finitely many occurrences of $\ireddownfin$.
The depth requirement in the definition of strong convergence
arises naturally in the rules~\eqref{rules:restrict}, 
in particular the middle rule 
pushes the activity to the arguments.
The fact that the rules~\eqref{rules:restrict} 
capture the infinitary rewriting relation $\ired$
is a consequence of a result due to~\cite{kenn:klop:slee:vrie:1995a}
which states that every strongly convergent rewrite sequence
contains only a finite number of steps at any depth $d \in \nat$, 
in particular only a finite number of root steps~$\rstep$.
Hence every strongly convergent reduction is of the form ${(\ireddownfin \relcomp \rstep)^*} \relcomp \ireddown$
as in the premise of the first rule,
where the steps $\ireddownfin$ are reductions of shorter length.

We conclude with an example of a TRS that allows for a rewrite sequence of length beyond $\omega$.

\begin{example}\label{ex:fab}
  We consider the term rewriting system with the following rules:
  {\abovedisplayskip.75ex 
   \belowdisplayskip.75ex
  \begin{align*}
    \fun{f}(x,x) &\to \fun{D} & 
    \fun{a} &\to \fun{C}(\fun{a}) &
    \fun{b} &\to \fun{C}(\fun{b}) \,.
  \end{align*}}%
  We then have $\fun{a} \ired \fun{C}^\omega$, that is, an infinite reduction from $\fun{a}$ to $\fun{C}^\omega$ in the limit:
  \begin{align*}
    \fun{a} \to \fun{C}(\fun{a}) \to \fun{C}(\fun{C}(\fun{a})) \to \fun{C}(\fun{C}(\fun{C}(\fun{a}))) \to \cdots 
    \to^\omega \fun{C}^\omega \,.
  \end{align*}
  Using the proof rules~\eqref{rules:restrict}, we can derive $\fun{a} \ired \fun{C}^\omega$ as shown in Figure~\ref{fig:aComega}.
  \noindent
\end{example}\vspace{-1.7ex}

  \begin{wrapfigure}{r}{5.2cm}
    \vspace{-3ex}
    \begin{framed}
    \vspace{-1.5ex}
    \begin{align*}
      \infer=
      {\fun{a} \ired \fun{C}^\omega}
      {
        \fun{a} \rstep \fun{C}(\fun{a})
        &
        \infer=
        {\fun{C}(\fun{a}) \ireddown \fun{C}^\omega}
        {\infer={\fun{a} \ired \fun{C}^\omega}
            {\makebox(0,0){
              \hspace{13mm}\begin{tikzpicture}[baseline=11ex]
              \draw [->,thick,dotted] (0,0) -- (0,1mm) to[out=90,in=100] (11mm,-1mm) to[out=-80,in=-20,looseness=1.4] (-0mm,-13mm);
              \end{tikzpicture}
            }}
        }
      }
    \end{align*}
    \vspace{-6ex}
    \end{framed}
    \vspace{-4ex}
    \caption{A reduction $\fun{a} \ired \fun{C}^\omega$.}
    \vspace{-3ex}
    \label{fig:aComega}
  \end{wrapfigure}
  The proof tree in Figure~\ref{fig:aComega} can be described as follows:
  We have an infinitary rewrite sequence from $\fun{a}$ to~$\fun{C}^\omega$
  since we have a root step from $\fun{a}$ to $\fun{C}(\fun{a})$, and
  an infinitary reduction below the root from $\fun{C}(\fun{a})$ to $\fun{C^\omega}$.
  The latter reduction $\fun{C}(\fun{a}) \ireddown \fun{C^\omega}$ is in turn witnessed
  by the infinitary rewrite sequence $\fun{a} \ired \fun{C}^\omega$ 
  on the direct subterms.

  We also have the following reduction, now of length $\omega+1$:
  \begin{align*}
    \fun{f}(\fun{a},\fun{b}) 
    \to \fun{f}(\fun{C}(\fun{a}),\fun{b}) 
    \to \fun{f}(\fun{C}(\fun{a}),\fun{C}(\fun{b}))
    \to \cdots \to^\omega \fun{f}(\fun{C}^\omega,\fun{C}^\omega)
    \to \fun{D} \,.
  \end{align*}
  That is, after an infinite rewrite sequence of length $\omega$, 
  we reach the limit term $\fun{f}(\fun{C}^\omega,\fun{C}^\omega)$,
  and we then continue with a rewrite step from $\fun{f}(\fun{C}^\omega,\fun{C}^\omega)$ to $\fun{D}$.

  \begin{wrapfigure}{r}{8.2cm}
    \vspace{-4ex}
    \begin{framed}
    \vspace{-1.5ex}
    \begin{align*}
      \infer=
      {\fun{f}(\fun{a},\fun{b}) \ired \fun{D}}
      {
        \infer=
        {\fun{f}(\fun{a},\fun{b}) \ireddownfin \fun{f}(\fun{C}^\omega,\fun{C}^\omega)}
        {
          \infer=
          {\fun{a} \ired \fun{C}^\omega}
          {\text{like Figure~\ref{fig:aComega}}}
          &
          \infer=
          {\fun{b} \ired \fun{C}^\omega}
          {\text{like Figure~\ref{fig:aComega}}}
        }
        & 
        \fun{f}(\fun{C}^\omega,\fun{C}^\omega) \rstep \fun{D}
      }
    \end{align*}
    \vspace{-5.5ex}
    \end{framed}
    \vspace{-4ex}
    \caption{A reduction $\fun{f}(\fun{a},\fun{b}) \ired \fun{D}$.}
    \vspace{-3ex}
    \label{fig:fab}
  \end{wrapfigure}

  Figure~\ref{fig:fab} shows how this rewrite sequence \mbox{$\fun{f}(\fun{a},\fun{b}) \ired \fun{D}$}
  can be derived in our setup.
  We note that the rewrite sequence $\fun{f}(\fun{a},\fun{b}) \ired \fun{D}$
  cannot be `compressed' to length $\omega$. 
  So there is no reduction $\fun{f}(\fun{a},\fun{b}) \to^{\le \omega} \fun{D}$.

\paragraph*{Related Work}
While 
a coinductive treatment of infinitary rewriting is not new~\cite{coqu:1996,joac:2004,endr:polo:2012b},
the previous approaches 
only capture rewrite sequences of length at most $\omega$. 
The coinductive framework that we present here captures all strongly
convergent rewrite sequences of arbitrary ordinal length.

From the topological perspective, various notions of infinitary rewriting 
and infinitary equational reasoning have been studied in~\cite{kahr:2013}.
%
The closure operator $S_E$ from \cite{kahr:2013} is closely related to our notion of infinitary equational reasoning $\ieq$.
The operator $S_E$ is defined by $S_E(R) = (S \circ E)^*(R)$ where 
$E(R)$ is the equivalence closure of $R$, and
$S(R)$ is the strongly convergent rewrite relation obtained from (single steps) $R$. 
Although defined in very different ways, 
the relations $S_E(\to)$ and $\ieq$ typically coincide.
In Section~\ref{sec:venn} we show that ${S_E(\to)} \subseteq {\ieq}$ for every rewrite system
and we give an example for which the inclusion is strict. 

Martijn Vermaat has formalized infinitary rewriting using metric convergence (in place of strong convergence) 
in the Coq proof assistant~\cite{verm:2010}, and proved that weakly orthogonal infinitary rewriting 
does not have the property $\mathrm{UN}$ of unique normal forms, see~\cite{endr:hend:grab:klop:oost:2014}.
While his formalization could be extended to strong convergence,
it remains to be investigated to what extent it can be used 
for the further development of the theory of infinitary rewriting.

Ketema and Simonsen~\cite{kete:simo:2013} 
introduce the notion of `computable infinite reductions'~\cite{kete:simo:2013},
where terms as well as reductions are computable,
and provide a Haskell implementation of the Compression Lemma for this notion of reduction.

\paragraph*{Outline}
In Section~\ref{sec:prelims} we introduce infinitary rewriting in the usual way
based on ordinals,
and with convergence at every limit ordinal.
Section~\ref{sec:coinduction} is a short explanation of (co)induction and fixed-point rules.
The two new definitions of infinitary rewriting~$\ired$ based on
mixing induction and coinduction, as well as their equivalence, are spelled out in Section~\ref{sec:itrs}.
Then, in Section~\ref{sec:equivalence},
we prove the equivalence of these new definitions of infinitary rewriting
with the standard definition.
In Section~\ref{sec:ieq} we present the above introduced relations~$\ieq$ and~$\ibi$
of infinitary equational reasoning and bi-infinite rewriting.
In Section~\ref{sec:venn} we compare the three relations $\ieq$, $\ibi$ and $\ired$.
As an application, we show in Section~\ref{sec:coq} that our framework
is suitable for formalizations in theorem provers.
We conclude in Section~\ref{sec:conclusion}.

\section{Preliminaries on Term Rewriting}\label{sec:prelims}

We give a brief introduction to infinitary rewriting.
For further reading on infinitary rewriting we refer to
\cite{klop:vrij:2005,tere:2003,bare:klop:2009,endr:hend:klop:2012},
for an introduction to finitary rewriting to
\cite{klop:1992,tere:2003,baad:nipk:1998,bare:1977}.

A \emph{signature $\Sigma$} is a set of symbols $f$ each having a fixed arity $\arity{f} \in \nat$.
Let $\avars$ be an infinite set of variables such that $\avars \cap \Sigma = \emptyset$.
The set $\iter{\Sigma}{\avars}$ of (finite and) \emph{infinite terms} 
over $\Sigma$ and $\avars$ is coinductively defined by the following grammar:
\begin{align*}
  T \coBNFis x \BNFor f(\underbrace{T,\ldots,T}_{\text{$\arity{f}$ times}}) \; \text{($x \in \avars$, $f \in \Sigma$)} \,.
\end{align*}
This means that $\iter{\Sigma}{\avars}$ is defined as the largest set $T$ such that 
for all $t \in T$, either $t \in \avars$ or $t = f(t_1,t_2,\ldots,t_n)$ for some $f \in \Sigma$ with $\arity{f} = n$
and $t_1,t_2,\ldots,t_n \in T$.
So the grammar rules may be applied an infinite number of times, and equality on the terms is bisimilarity.
See further Section~\ref{sec:coinduction} for a brief introduction to coinduction.

We write $\id$ for the identity relation on terms, $\id \defd \{\pair{s}{s} \mid s \in \iter{\Sigma}{\avars}\}$.

\begin{mremark}
  Alternatively, the set $\iter{\Sigma}{\avars}$ arises from the set of finite terms, $\ter{\Sigma}{\avars}$, 
  by metric completion,
  using the well-known distance function $\smetric$ 
  defined by
  $\metric{t}{s} = 2^{-n}$ if the $n$-th level of the terms $t,s \in \ter{\Sigma}{\avars}$ (viewed as labeled trees) 
  is the first level where a difference appears, 
  in case $t$ and $s$ are not identical; furthermore, $\metric{t}{t} = 0$.
  It is standard that this construction yields $\pair{\ter{\Sigma}{\avars}}{\smetric}$ as a metric space. 
  Now infinite terms are obtained by taking the completion of this metric space, 
  and they are represented by infinite trees. 
  We will refer to the complete metric space arising in this way as $\pair{\iter{\Sigma}{\avars}}{\smetric}$, 
  where $\iter{\Sigma}{\avars}$ is the set of finite and infinite terms over~$\Sigma$.
\end{mremark}

Let $t \in \iter{\Sigma}{\avars}$ be a finite or infinite term.
The set of \emph{positions $\pos{t}\subseteq \nat^*$ of $t$} is
defined by: $\posemp \in \pos{t}$, and 
$i p \in \pos{t}$ whenever $t = f(t_1,\ldots,t_n)$ 
with $1 \le i \le n$ and $p \in \pos{t_i}$.
For $p \in \pos{t}$, the \emph{subterm} $\subtrm{t}{p}$ of $t$ at position $p$
is defined by
$\subtrm{t}{\posemp} = t$ and
$\subtrm{f(t_1,\ldots,t_n)}{ip} = \subtrm{t_i}{p}$.
The set of \emph{variables $\vars{t}\subseteq \avars$ of $t$} is %
$\vars{t} = \{x \in \avars \mid \exists \, p\in \pos{t}.\,\subtrm{t}{p} = x\}$.

A \emph{substitution $\asubst$} is a map $\asubst : \avars \to \iter{\Sigma}{\avars}$;
its domain 
is extended to
$\iter{\Sigma}{\avars}$ by corecursion: %
$\asubst(f(t_1,\ldots,t_n)) = f(\asubst(t_1),\ldots,\asubst(t_n))$.
For a term~$t$ and a substitution~$\asubst$, we write $t\sigma$ for $\sigma(t)$.
We write $x \mapsto s$ for the substitution defined by
$\asubst(x) = s$ and $\asubst(y) = y$ for all $y \ne x$.
Let $\hole$ be a fresh variable.
A \emph{context} $C$ is a term $\iter{\Sigma}{\avars \cup \{\hole\}}$
containing precisely one occurrence of %
$\hole$.
For contexts $C$ and terms $s$ we write $C[s]$ for $C(\hole \mapsto s)$.

A \emph{rewrite rule $\ell \to r$} over $\Sigma$ and $\avars$ is a pair 
$(\ell,r)$
of terms $\ell,r\in\iter{\Sigma}{\avars}$ such that the left-hand side $\ell$ is not a variable ($\ell \not\in \avars$),
and all variables in the right-hand side $r$ occur in $\ell$, $\vars{r} \subseteq \vars{\ell}$.
Note that we require neither the left-hand side nor the right-hand side of a rule
to be finite.

A \emph{term rewriting system (TRS) $\atrs$} over $\Sigma$ and $\avars$
is a set of rewrite rules over $\Sigma$ and $\avars$.
A TRS induces a rewrite relation on the set of terms as follows.
For $\apos \in \nat^\ast$ we define ${\rerat{\atrs}{\apos}} \subseteq \iter{\Sigma}{\avars} \times \iter{\Sigma}{\avars}$, 
a \emph{rewrite step at position $\apos$}, by
$
  C[\ell\sigma] \rerat{\atrs}{\apos} C[r\sigma]
$
if $C$ is a context with $\subtrm{C}{\apos} = \hole$,\; $\ell \to r \in \atrs$, and $\sigma : \avars \to \iter{\Sigma}{\avars}$.
We write $\rstep$ for \emph{root steps}, %
${\rstep} = \{\,(\ell\sigma,r\sigma) \mid \ell \to r \in \atrs,\; \text{$\sigma$ a substitution}\,\}$.
We write $s \to_{\atrs} t$ if $s \rerat{\atrs}{\apos} t$ for some $\apos\in\nat^\ast$.
A \emph{normal form} is a term without a redex occurrence,
that is, a term that is not of the form $C[\ell\sigma]$ for some context $C$, 
rule $\ell \to r\in \atrs$ and substitution $\sigma$.

A natural consequence of this construction is %
the notion of \emph{weak convergence}: 
we say that $t_0 \to t_1 \to t_2 \to \cdots$ is an infinite reduction sequence with limit $t$, 
if $t$ is the limit of the sequence $t_0,t_1,t_2, \ldots$ in the usual sense of metric convergence. 
We use \emph{strong convergence}, which in addition to weak convergence, 
requires that the depth of the redexes contracted in the 
successive steps tends to infinity when approaching a
limit ordinal from below.
So this rules out the possibility that the action of redex contraction stays confined at the top, 
or stagnates at some finite level of depth. 

\begin{definition}\label{def:itrs:standard}
  A \emph{transfinite rewrite sequence} (of ordinal length $\alpha$)
  is a sequence of rewrite steps 
  $(t_{\beta} \rerat{\atrs}{\apos_{\beta}} t_{\beta+1})_{\beta < \alpha}$
  such that for every limit ordinal $\lambda < \alpha$ we have that 
  if $\beta$ approaches $\lambda$ from below, then
  \begin{enumerate}[label=(\roman*)]
    \item\label{item:distance}
      the distance $\metric{t_\beta}{t_\lambda}$ tends to $0$ 
      and, moreover,
    \item\label{item:depth}
      the depth of the rewrite action, i.e., 
      the length of the position $\apos_\beta$, 
      tends to infinity.
  \end{enumerate}
  The sequence is called \emph{strongly convergent} 
  if $\alpha$ is a successor ordinal, 
  or there exists a term $t_\alpha$ such that
  the conditions~\ref{item:distance} and~\ref{item:depth}
  are fulfilled for every limit ordinal $\lambda \leq \alpha$;
  we then write $t_0\iredord t_\alpha$.
  The subscript $\mit{ord}$ is used in order to distinguish 
  $\iredord$ from the equivalent %
  relation $\ired$ 
  as defined in Definition~\ref{def:ired:fixedpoint}. 
  We sometimes write $t_0\redord^{\alpha} t_\alpha$ 
  to explicitly indicate the length $\alpha$ of the sequence.
  The sequence is called \emph{divergent} if it is not strongly convergent.
\end{definition}

There are several reasons why strong convergence is beneficial; 
the foremost being that in this way we can define the notion of \emph{descendant} 
(also \emph{residual}) over limit ordinals. 
Also the well-known Parallel Moves Lemma
and the Compression Lemma
fail for weak convergence, see~\cite{simo:2004} and \cite{ders:kapl:plai:1991} respectively.

\section{(Co)induction and Fixed Points}\label{sec:coinduction}

We briefly introduce the relevant concepts from
(co)algebra and (co)induction that will be used later throughout this
paper. 
For a more thorough introduction, we refer to \cite{jaco:rutt:2011}.
There will be two main points where coinduction will play a role, in the definition of terms and in the definition 
of term rewriting. 

Terms are usually defined with respect to a type constructor~$F$. 
For instance, consider the type of lists with elements in a given set $A$, given in a functional programming style:
\begin{verbatim}
  type List a = Nil | Cons a (List a)
\end{verbatim}
The above grammar corresponds to the type constructor
$F(X) = 1 + A \times X$ where the $1$ is used as a placeholder for the empty list {\tt Nil} 
and the second component represents the {\tt Cons} constructor.
Such a grammar can be interpreted in two ways: 
The \emph{inductive} interpretation yields as terms the set of finite lists,
and corresponds to the \emph{least fixed point} of~$F$.
The \emph{coinductive} interpretation yields as terms
the set of all finite or infinite lists,
and corresponds to the \emph{greatest fixed point} of~$F$.
More generally, the inductive interpretation of a type constructor 
yields finite terms (with well-founded syntax trees), and
dually, the coinductive interpretation 
yields possibly infinite terms.
For readers familiar with the categorical definitions of algebras
and coalgebras, these two interpretations amount to defining
finite terms as the \emph{initial $F$-algebra}, 
and possibly infinite terms as the \emph{final $F$-coalgebra}.

Formally, term rewriting is a relation on a set $T$ of terms,
and hence an element of the complete lattice $\tlat := \Pow(T \times T)$, 
the powerset of $T \times T$.
Relations on terms can thus be defined using least and greatest fixed points
of monotone operators on $\tlat$.
In this setting, an inductively defined relation is 
a least fixed point $\lfp{X}{F(X)}$ of a monotone $F : \tlat \to \tlat$.
Dually, a coinductively defined relation is
a greatest fixed point $\gfp{X}{F(X})$ of a monotone $F : \tlat \to \tlat$.
Coinduction, and similarly induction, can be formulated as proof rules:
\begin{gather}
  \begin{aligned}
    \frac{X \leq F(X)}{X \leq \gfp{Y}{F(Y)}}(\nu\text{-rule}) 
    &&\qquad\qquad\qquad&&
    \frac{F(X) \leq X}{\lfp{Y}{F(Y)} \leq X}(\mu\text{-rule})
  \end{aligned}
  \label{eq:coind-ind-rules}
\end{gather}
These rules express the fact that $\gfp{Y}{F(Y)}$ is the greatest post-fixed point of $F$,
and $\lfp{Y}{F(Y)}$ is the least pre-fixed point of $F$.

\section{New Definitions of Infinitary Term Rewriting}\label{sec:itrs}

We present two new definitions of infinitary rewriting \mbox{$s \ired t$},
based on mixing induction and coinduction, and prove their equivalence.
In Section~\ref{sec:equivalence} we show they are equivalent to the standard definition based on ordinals.
We summarize the definitions:
\begin{enumerate}[label=\emph{\Alph*}.]
  \item \emph{Derivation Rules.}
        First, we define $s \ired t$ via a syntactic restriction on the proof trees 
        that arise from the coinductive rules~\eqref{rules:restrict}.
        The restriction excludes all proof trees that contain ascending paths
        with an infinite number of marked symbols.
  \medskip
  \item \emph{Mixed Induction and Coinduction.}
        Second, we define $s \ired t$ based on mutually mixing induction and coinduction,
        that is, least fixed points $\mu$ and greatest fixed points $\nu$.
\end{enumerate}
\noindent
In contrast to previous coinductive definitions~\cite{coqu:1996,joac:2004,endr:polo:2012b}, 
the setup proposed here captures all strongly convergent rewrite sequences (of arbitrary ordinal length).

Throughout this section, we fix a signature $\Sigma$ and a term
rewriting system $\atrs$ over $\Sigma$. 
We also abbreviate $T \defd \iter{\Sigma}{\avars}$.
\begin{notation}\label{not:transitivity}
  Instead of introducing separate derivation rules for transitivity,
  we write a reduction of the form
  $s_0 \rightsquigarrow s_1 \rightsquigarrow \cdots \rightsquigarrow s_n$
  as a sequence of single steps: 
  {\abovedisplayskip2ex 
   \belowdisplayskip.75ex
  \begin{align*}
    \infer
    {\text{conclusion}}
    {s_0 \rightsquigarrow s_1 \quad s_1 \rightsquigarrow s_2 \quad\cdots\quad s_{n-1} \rightsquigarrow s_n}
  \end{align*}}%
  \noindent
  This allows us to write the subproof immediately above a single step.
\end{notation}

\begin{definition}\label{def:lifting}
  For a relation $R \subseteq T \times T$ %
  we define its \emph{lifting $\down{R}$} by
  \begin{align*}
    \down{R} \;\defd\; \{\,\pair{f(s_1,\ldots,s_n)}{\,f(t_1,\ldots,t_n)} \mid f \in \Sigma ,\, \arity{f} = n\,, s_1 \mathrel{R} t_1,\ldots,s_n \mathrel{R} t_n\,\}
                \,\cup\, \id \,.
  \end{align*}
\end{definition}

\subsection{Derivation Rules}

\begin{definition}\label{def:ired:restrict}
  We define the relation ${\ired} \subset T \times T$ as follows.
  We have $s \ired t$ if there exists a 
  (finite or infinite) proof tree $\delta$ deriving $s \ired t$ using the following five rules:
  \begin{align*}
    \infer=[\rsplit]
    {s \ired t}
    {s \mathrel{(\rstep \cup \ireddownfin)^*} \relcomp \ireddown t}
    &&  
    \infer=[\rlift]
    {f(s_1,s_2,\ldots,s_n) \ireddownmfin f(t_1,t_2,\ldots,t_n)}
    {s_1 \ired t_1 & \cdots & s_n \ired t_n}
    &&
    \infer=[\rid]
    {s \ireddownmfin s}
    {}
  \end{align*}
  such that $\delta$ does not contain an infinite nesting of $\ireddownfin$,
  that is, such that there exists no path ascending 
  through the proof tree that meets an infinite number of symbols $\ireddownfin$. 
  The symbol $\ireddownmfin$ stands for $\ireddown$ or $\ireddownfin$;
  so the second rule is an abbreviation for two rules; similarly for the third rule.
\end{definition}

We give some intuition for the rules in Definition~\ref{def:ired:restrict}.
The relations $\ireddownfin$ and $\ireddown$ are infinitary reductions below the root.
We use $\ireddownfin$ for constructing parts of the prefix (between root steps), and
$\ireddown$ for constructing a suffix of the reduction that we are defining.
When thinking of ordinal indexed rewrite sequences $\sigma$, 
a suffix of $\sigma$ can have length equal to $\sigma$,
while the length of every prefix of $\sigma$ must be strictly smaller than the length of $\sigma$.
The five rules (\rsplit, and the two versions of \rlift\ and \rid) can be interpreted as follows:
\begin{enumerate}
  \item 
    The \rsplit-rule:
    the term $s$ rewrites infinitarily to $t$, $s \ired t$, 
    if $s$ rewrites to $t$ using a finite sequence of (a) root steps,
    and (b) infinitary reductions $\ireddown$ below the root
    (where infinitary reductions preceding root steps must be shorter than the derived reduction).
    \smallskip
  \item 
    The \rlift-rules: 
    the term $s$ rewrites infinitarily to $t$ below the root, $s \ireddownmfin t$, %
    if the terms are of the shape $s = f(s_1,s_2,\ldots,s_n)$ and $t = f(t_1,t_2,\ldots,t_n)$
    and there exist reductions %
    on the arguments:
    $s_1 \ired t_1$, \ldots, $s_n \ired t_n$.
    \smallskip
  \item 
    The \rid-rules allow for the rewrite relations~$\ireddownmfin$ %
    to be reflexive,
    and this in turn yields reflexivity of $\ired$.
    For variable-free terms, reflexivity can already be derived using the other %
    rules.
    For terms with variables, this %
    rule is needed 
    (unless we treat variables as constant symbols).
\end{enumerate}
For an example of a proof tree, 
we refer to Example~\ref{ex:fab} in the introduction. %

\subsection{Mixed Induction and Coinduction}

The next definition is based on mixing induction and coinduction. 
The inductive part is used to model the restriction 
to finite nesting of $\ireddownfin$ in the proofs in Definition~\ref{def:ired:restrict}.
The induction corresponds to a least fixed point $\slfp$,
while a coinductive rule to a greatest fixed point $\sgfp$.

\begin{definition}\label{def:ired:fixedpoint}
  We define the relation ${\ired} \subseteq T \times T$ by
  \begin{align*}
      {\ired} \;\;\defd\;\; \lfp{R}{\gfp{S}{(\rstep \cup \mathrel{\down{R}})^*\relcomp \down{S}}} \,.
  \end{align*}
\end{definition}

We argue why ${\ired}$ is well-defined.
Let $L \defd \Pow(T\times T)$ be the set of all relations on terms.
Define functions $G : L \times L \to L$ and $F : L \to L$ by
\begin{align}
  G(R,S) \defd (\rstep \cup \mathrel{\down{R}})^*\relcomp \down{S}
    \quad\text{ and }\quad
  F(R) \defd \gfp{S}{G(R,S)} = \gfp{S}{(\rstep \cup \mathrel{\down{R}})^*\relcomp \down{S}} \,.
  \label{eq:G:F}
\end{align}
Then we have  %
${\ired} \;=\; \lfp{R}{F(R)} \;=\; \lfp{R}{\gfp{S}{G(R,S)}} \;=\; 
\lfp{R}{\gfp{S}{(\rstep \cup \mathrel{\down{R}})^*\relcomp \down{S}}}$.
It can easily be verified that $F$ and $G$ are monotone (in all their arguments).
Recall that a function~$H$ over sets is monotone 
if $X \subseteq Y$ implies $H(\ldots,X,\ldots) \subseteq H(\ldots,Y,\ldots)$.
Hence $F$ and $G$ have unique least and greatest fixed points.

\subsection{Equivalence}

We show equivalence of Definitions~\ref{def:ired:restrict} and~\ref{def:ired:fixedpoint}.
Intuitively, the $\slfp R$ in the fixed point definition 
corresponds to the nesting restriction in the definition using derivation rules. 
If one thinks of Definition~\ref{def:ired:fixedpoint} as $\lfp{R}{F(R)}$
with $F(R) = \gfp{S}{G(R,S)}$
(see equation~\eqref{eq:G:F}), then $F^{n+1}(\emptyset)$ 
are all infinite rewrite sequences that can be derived 
using proof trees where the nesting depth of the marked symbol $\ireddownfin$ is at most $n$. 

To avoid confusion we write $\ired_{\nest}$ for the relation $\ired$ defined in Definition~\ref{def:ired:restrict},
and $\ired_{\fp}$ for the relation $\ired$ defined in Definition~\ref{def:ired:fixedpoint}.
We show ${\ired_{\nest}} = {\ired_{\fp}}$.
Definition~\ref{def:ired:restrict} requires that the nesting structure of $\ireddownfin_\nest$
in proof trees is well-founded. As a consequence, we can associate to every proof tree
a (countable) ordinal that allows to embed the nesting structure in an order-preserving way.
We use $\omega_1$ to denote the first uncountable ordinal,
and we view ordinals as the set of all smaller ordinals
(then the elements of $\omega_1$ are all countable ordinals).

\begin{definition}
  Let $\delta$ be a proof tree as in Definition~\ref{def:ired:restrict},
  and let $\alpha$ be an ordinal.
  An \emph{$\alpha$-labeling of $\delta$} 
  is a labeling of all symbols $\ireddownfin_\nest$ in $\delta$ with elements from $\alpha$
  such that
  each label is strictly greater than all labels occurring in the subtrees (all labels above).
\end{definition}

\begin{lemma}\label{lem:nest}
  Every proof tree as in Definition~\ref{def:ired:restrict}
  has an $\alpha$-labeling for some $\alpha \in \omega_1$.
\end{lemma}

\begin{proof}
  Let $\delta$ be a proof tree and 
  let $L(\delta)$ be the set positions of symbols $\ireddownfin_\nest$ in $t$.
  For positions $p,q \in L(\delta)$ we write $p < q$ if $p$ is a strict prefix of $q$.
  Then we have that $>$ is well-founded, that is, 
  there is no infinite sequence $p_0 < p_1 < p_2 < \cdots$ with $p_i \in L(\delta)$ ($i \ge 0$)
  as a consequence of the nesting restriction on $\ireddownfin_\nest$. 
  The the extension of this well-founded order on $L(t)$ 
  to a total, well-founded order is isomorphic to an ordinal $\alpha$,
  and $\alpha < \omega_1$ since $L(t)$ is countable. 
\end{proof}

\begin{definition}
  Let $\delta$ be a proof tree as in Definition~\ref{def:ired:restrict}.
  We define the \emph{nesting depth} of $\delta$ as 
  the least ordinal $\alpha \in \omega_1$ such that $\delta$ admits an $\alpha$-labeling.
  For every $\alpha \le \omega_1$, we define a relation
  ${\ired_{\alpha,\nest}}  \subseteq {\ired_{\nest}}$
  as follows:
  $s \ired_{\alpha,\nest} t$ whenever $s \ired_\nest t$
  can be derived using a proof with nesting depth $< \alpha$.
  Likewise we define relations
  ${\ireddown_{\alpha,\nest}}$ and
  ${\ireddownfin_{\alpha,\nest}}$\,.
\end{definition}

As a direct consequence of Lemma~\ref{lem:nest} we have:
\begin{corollary}\label{cor:nest:omega1}
  We have ${\ired_{\omega_1,\nest}} = {\ired_\nest}$.
\end{corollary}

\begin{theorem}\label{thm:equiv:der:fp}
  Definitions~\ref{def:ired:restrict} and~\ref{def:ired:fixedpoint} 
  define the same relation, ${\ired_{\nest}} = {\ired_{\fp}}$. %
\end{theorem}

\begin{proof}
  We begin with ${\ired_{\fp}} \subseteq {\ired_{\nest}}$.
  Recall that $F(\ired_\nest)$ is the greatest fixed point of $G(\ired_\nest,\_)$, see~\eqref{eq:G:F}.
  Also, we have ${\ireddown_{\nest}} = {\ireddownfin_{\nest}} = \down{\ired_{\nest}\vphantom{i}}$, 
  and hence
  \begin{align}
    F({\ired_{\nest}}) &= (\rstep \cup \mathrel{\down{\ired_{\nest}\vphantom{i}}})^* \relcomp \down{F({\ired_{\nest}})} 
    = (\rstep \cup \mathrel{\ireddownfin_{\nest}})^* \relcomp \down{F({\ired_{\nest}})}\\
    \down{F({\ired_{\nest}})} &= \id \cup \{\,\pair{f(\vec{s})}{\,f(\vec{t})} \mid \vec{s} \,\mathrel{F(\ired_{\nest})}\, \vec{t}\,\}
    \label{eq:id:or:not}
  \end{align}
  where $\vec{s}$, $\vec{t}$ abbreviate $s_1,\ldots,s_n$ and $t_1,\ldots,t_n$, respectively,
  and we write $\vec{s} \mathrel{R} \vec{t}$
  if we have $s_1 \mathrel{R} t_1,\ldots,s_n \mathrel{R} t_n$.
  Employing the $\mu$-rule from~\eqref{eq:coind-ind-rules},
  it suffices to show that $F({\ired_{\nest}}) \subseteq {\ired_{\nest}}$.
  Assume $\pair{s}{t} \in F({\ired_{\nest}})$.
  Then $\pair{s}{t} \in (\rstep \cup \mathrel{\ireddownfin_{\nest}})^* \relcomp \down{F({\ired_{\nest}})}$.
  Then there exists $s'$ such that $s \mathrel{(\rstep \cup \mathrel{\ireddownfin_{\nest}})^*} s'$
  and $s' \mathrel{\down{F({\ired_{\nest}})}} t$.
  Now we distinguish cases according to~\eqref{eq:id:or:not}:
  \begin{align*}
    \infer=[\rsplit] {s \ired t} {s \mathrel{(\rstep \cup \mathrel{\ireddownfin_{\nest}})^*} t & 
      \infer=[\rid] {t \ireddown t } {} }
    &&
    \infer=[\rsplit] {s \ired t} {s \mathrel{(\rstep \cup \mathrel{\ireddownfin_{\nest}})^*} s' & 
      \infer=[\rlift] {s' \ireddown t } {T_1 & \cdots & T_n} }
  \end{align*}
  Here, for $i \in \{1,\ldots,n\}$, $T_i$ is the proof tree for $s_i \ired t_i$
  obtained from $s_i \mathrel{F(\ired_{\nest})} t_i$ by corecursively applying the same procedure.

  Next we show that ${\ired_{\nest}} \subseteq {\ired_{\fp}}$.
  By Corollary~\ref{cor:nest:omega1} %
  it suffices to show ${\ired_{\omega_1,\nest}} \subseteq {\ired_{\fp}}$.
  We prove by well-founded induction on $\alpha \le \omega_1$ that
  ${\ired_{\alpha,\nest}} \subseteq {\ired_{\fp}}$.
  Since $\ired_{\fp}$ is a fixed point of $F$,
  we obtain ${\ired_{\fp}} = F(\ired_{\fp})$, and since $F(\ired_{\fp})$ 
  is a greatest fixed point,
  using the $\nu$-rule from~\eqref{eq:coind-ind-rules},
  it suffices to show that $(*)$ ${\ired_{\alpha,\nest}} \subseteq G(\ired_{\fp},\ired_{\alpha,\nest})$.
  Thus assume that $s \ired_{\alpha,\nest} t$,
  and let $\delta$ be a proof tree of nesting depth $\le \alpha$ deriving $s \ired_{\alpha,\nest} t$.
  The only possibility to derive $s \ired_{\nest} t$ is an application of the \rsplit-rule
  with the premise $s \mathrel{(\rstep \cup \ireddownfin_{\nest})^*} \relcomp \ireddown_{\nest} t$.
  Since $s \ired_{\alpha,\nest} t$, we have
  $s \mathrel{(\rstep \cup \ireddownfin_{\alpha,\nest})^*} \relcomp \ireddown_{\alpha,\nest} t$.
  Let $\tau$ be one of the steps $\ireddownfin_{\alpha,\nest}$ displayed in the premise.
  Let $u$ be the source of $\tau$ and $v$ the target,
  so $\tau : u \ireddownfin_{\alpha,\nest} v$.
  The step $\tau$ is derived either via the \rid-rule or the \rlift-rule.
  The case of the \rid-rule is not interesting since we then can drop $\tau$ from the premise.
  Thus let the step $\tau$ be derived using the \rlift-rule.
  Then the terms $u,v$ are of form $u = f(u_1,\ldots,u_n)$ and $v = f(v_1,\ldots,v_n)$
  and for every $1 \le i \le n$ we have $u_i \ired_{\beta,\nest} v_i$ for some $\beta < \alpha$.
  Thus by induction hypothesis we obtain $u_i \ired_{\fp} v_i$ for every $1 \le i \le n$,
  and consequently $u \mathrel{\down{\ired_{\fp}\vphantom{i}}} v$.
  We then have $s \mathrel{(\rstep \cup \mathrel{\down{\ired_{\fp}\vphantom{i}})^*}} \relcomp \ireddown_{\alpha,\nest} t$,
  and hence $s \mathrel{G(\ired_{\fp},\ired_{\alpha,\nest})} t$.
  This concludes the proof.
\end{proof}

\section{Equivalence with the Standard Definition}\label{sec:equivalence}

In this section we prove the equivalence of the coinductively defined 
infinitary rewrite relations~$\ired$ from
Definitions~\ref{def:ired:restrict} (and~\ref{def:ired:fixedpoint})
with the standard definition
based on ordinal length rewrite sequences with metric and strong convergence at every limit ordinal
(Definition~\ref{def:itrs:standard}).
The crucial observation is the following theorem from~\cite{klop:vrij:2005}:
\begin{theorem}[\normalfont{Theorem 2 of~\cite{klop:vrij:2005}}]\label{thm:finite}
  A transfinite reduction is divergent if and only if for some $n \in \nat$
  there are infinitely many steps at depth $n$.
\end{theorem}

We are now ready to prove the equivalence of both notions:
\begin{theorem}\label{thm:ired:equiv}
  We have ${\ired} = {\iredord}$.
\end{theorem}

\begin{proof}
  We write $\iredorddown$ to denote a reduction $\iredord$ without root steps,
  and we write $\redord^\alpha$ and $\redorddown^\alpha$ to indicate the ordinal length $\alpha$.

  We begin with the direction ${\iredord} \subseteq {\ired}$.
  We define a function $\smktree$ (and $\smktreedownmfin$) 
  by guarded corecursion~\cite{coqu:1994},
  mapping rewrite sequences $s \redord^\alpha t$
  (and $s \redorddown^\alpha t$)
  to infinite proof trees derived using the rules from Definition~\ref{def:ired:restrict}.
  This means that every recursive call produces a constructor, contributing to the construction of the infinite tree.
  Note that the arguments of $\smktree$ (and $\smktreedownmfin$)
  are not required to be structurally decreasing.

  We do case distinction on the ordinal $\alpha$.
  If $\alpha = 0$, then $t = s$ and we define %
  \begin{align*}
    \mktree{s \redord^0 s} 
    \;\;=\;\; \infer=[\rsplit] {s \ired s}{\mktreedown{s \redorddown^0 s}}
    &&
    \mktreedownmfin{s \redorddown^0 s} 
    \;\;=\;\; \infer=[\rid] {s \ireddownmfin s}{}
  \end{align*}

  If $\alpha > 0$, then, by Theorem~\ref{thm:finite} the rewrite sequence $s \redord^\alpha t$
  contains only a finite number of root steps.
  As a consequence, it is of the form:
  \begin{align*}
    s = s_0 \sto s_1 \cdots \sto s_{n-1} \sto s_n = t
  \end{align*}
  where for every $i \in \{0,\ldots,n-1\}$, 
  $s_i \sto s_{i+1}$
  is either a root step $s_i \rstep s_{i+1}$,
  or an infinite reduction below the root $s_i \redorddown^{\le \alpha} s_{i+1}$ 
  where $s_i \redorddown^{< \alpha} s_{i+1}$ if $i < n-1$.
  In the latter case, 
  the length of   $s_i \redorddown s_{i+1}$ is smaller than $\alpha$
  because every strict prefix must be shorter than the sequence itself.
  We define
  \begin{align*}
    \mktree{s \redord^\alpha t}
    \;\;=\;\; 
    \infer=[\rsplit] {s \ired t} { T_0 & T_1 & \cdots & T_{n-1} }
  \end{align*}
  where, for $0 \le i < n$,
  \begin{align*}
    T_i = 
    \begin{cases}
      s_i \rstep s_{i+1} & \text{if $s_i \sto s_{i+1}$ is a root step,} \\
      \mktreedownfin{s_i \redorddown^{\beta} s_{i+1}} & \text{if $i < n-1$ and $s_i \redorddown^{\beta} s_{i+1}$ for some $\beta < \alpha$,} \\
      \mktreedown{s_i \redorddown^{\beta} s_{i+1}} & \text{if $i = n-1$ and $s_i \redorddown^{\beta} s_{i+1}$ for some $\beta \le \alpha$.}
    \end{cases}
  \end{align*} 
  
  For rewrite sequences $s \redorddown^\alpha t$ with $\alpha > 0$
  we have that $s = f(s_1,\ldots,s_n)$ and $t = f(t_1,\ldots,t_n)$ for some $f \in \Sigma$ of arity~$n$
  and terms $s_1,\ldots,s_n,t_1,\ldots,t_n \in \iter{\Sigma}{\avars}$,
  and there is a rewrite sequence $s_i \redord^{\le \alpha} t_i$ for every $i$ with $1 \le i \le n$. 
  We define the two rules:
  \begin{align*}
    \mktreedownmfin{s \redorddown^\alpha t}
    \;\;=\;\;
    \infer=[\rlift] {s \ireddownmfin t}{\mktree{s_1 \redord^{\le \alpha} t_1} & \cdots & \mktree{s_n \redord^{\le \alpha} t_n}}
  \end{align*}
  
  The obtained proof tree $\mktree{s \redord^\alpha t}$ derives $s \ired t$.
  To see that the requirement that there is no ascending path through this tree 
  containing an infinite 
  number of symbols $\ireddownfin$ is fulfilled,
  we note the following.
  The symbol $\ireddownfin$ is produced by $\mktreedownfin{s \redorddown^\beta t}$
  which is invoked in $\mktree{s \redord^\alpha t}$ for a $\beta$ that is strictly smaller than $\alpha$.
  By well-foundedness of $<$ on ordinals, no such path exists.

  We now show ${\ired} \subseteq {\iredord}$.
  We prove by well-founded induction on $\alpha \le \omega_1$ that
  ${\ired_{\alpha}} \subseteq {\iredord}$.
  This suffices since ${\ired} = {\ired_{\omega_1}}$.
  Let $\alpha \le \omega_1$ and assume that $s \ired_{\alpha} t$.
  Let $\delta$ be a proof tree of nesting depth $< \alpha$ witnessing $s \ired_{\alpha} t$.
  The only possibility to derive $s \ired t$ is an application of the \rsplit-rule
  with the premise $s \mathrel{(\rstep \cup \ireddownfin)^*} \relcomp \ireddown t$.
  Since $s \ired_{\alpha} t$, we have
  $s \mathrel{(\rstep \cup \ireddownfin_{\alpha})^*} \relcomp \ireddown_{\alpha} t$.
  By induction hypothesis we have 
  $s \mathrel{(\rstep \cup \iredord)^*} \relcomp \ireddown_{\alpha} t$,
  and thus
  $s \iredord \relcomp \ireddown_{\alpha} t$.
  We have ${\ireddown_{\alpha}} = {\down{\ired_{\alpha}\vphantom{i}}}$, and consequently
  $s \iredord s_1 \mathrel{\down{\ired_{\alpha}\vphantom{i}}} t$ for some term $s_1$.  
  Repeating this argument on
  $s_1 \mathrel{\down{\ired_{\alpha}\vphantom{i}}} t$, we get 
  $s \iredord s_1 \mathrel{\down{\iredord\vphantom{i}}} s_2 \mathrel{\down{\down{\ired_{\alpha}\vphantom{i}}}} t$.  
  After $n$ iterations, we obtain
  \begin{align*}
    s \iredord s_1 
    \mathrel{\down{\iredord\vphantom{i}}} s_2 
    \mathrel{\down{\down{\iredord\vphantom{i}}}} s_3
    \mathrel{\down{\down{\down{\iredord\vphantom{i}}}}} s_4
    \cdots \mathrel{({\ired_{\alpha}})^{-(n-1)}} s_n
    \mathrel{({\ired_{\alpha}})^{-n}} t
  \end{align*}
  where $({\ired_{\alpha}})^{-n}$
  denotes the $n$th iteration of $x \mapsto \down{x}$ on $\ired_{\alpha}$.
  
  Clearly, the limit of $\{s_n\}$ is $t$.  Furthermore, each of the reductions $s_n \iredord s_{n+1}$ 
  are strongly convergent and take place at depth greater than or equal to $n$.
  Thus, the infinite concatenation of these reductions yields a strongly convergent reduction from $s$ to $t$
  (there is only a finite number of rewrite steps at every depth $n$).
\end{proof}

\section{Infinitary Equational Reasoning and Bi-Infinite Rewriting}\label{sec:ieq}

\subsection{Infinitary Equational Reasoning}%

\begin{samepage}
\begin{definition}\label{def:ieq:rules}
  Let $\aes$ be a TRS over $\Sigma$, and let $T = \iter{\Sigma}{\avars}$.
  We define \emph{infinitary equational reasoning} as the relation~%
  ${=^\infty} \subseteq T \times T$
  by the %
  mutually coinductive rules:
  \begin{align*}
    \infer=
    {s \ieq t}
    {s \mathrel{(\leftarrow_{\varepsilon} \cup \rstep \cup \ieqdown)^*} t}
    &&  
    \infer=
    {f(s_1,s_2,\ldots,s_n) \ieqdown f(t_1,t_2,\ldots,t_n)}
    {s_1 \ieq t_1 & \cdots & s_n \ieq t_n}
  \end{align*}
  where ${\ieqdown} \subseteq T \times T$ stands for infinitary equational reasoning below the root.
\end{definition}
\end{samepage}

Note that, in comparison with the rules~\eqref{rules:intro:ieq} for $\ieq$ from the introduction, 
we now have used $\leftarrow_\varepsilon \cup \rstep$ instead of $=_\aes$.
It is not difficult to see that this gives rise to the same relation.
The reason is that we can `push' non-root rewriting steps $=_\aes$ 
into the arguments of~$\ieqdown$.

\begin{example}\label{ex:ieq}
  Let $\aes$ be a TRS consisting of the following rules:
  \begin{align*}
    \fun{a} & \to \fun{f}(\fun{a})  &
    \fun{b} & \to \fun{f}(\fun{b}) &
    \fun{C}(\fun{b}) & \to \fun{C}(\fun{C}(\fun{a})) \,.
  \end{align*}
  Then we have $\fun{a} \ieq \fun{b}$ as derived in Figure~\ref{fig:ieq:fagb} (top),
  and $\fun{C}(\fun{a}) \ieq \fun{C}^\omega$ as in Figure~\ref{fig:ieq:fagb} (bottom).
  \begin{figure}[ht]
    \begin{gather*}
    \infer=
    { \fun{a} \ieq \fun{b} }
    {
      \fun{a} \to_{\varepsilon} \fun{f}(\fun{a})
      &
      \infer=
      {\fun{f}(\fun{a}) \ieqdown \fun{f}^\omega}
      {
        \infer=
        {\fun{a} \ieq \fun{f}^\omega}
        {
          \fun{a} \to_{\varepsilon} \fun{f}(\fun{a})
          &
          \infer=
          {\fun{f}(\fun{a}) \ieqdown \fun{f}^\omega}
          {{\makebox(0,0){
                \hspace{13mm}\begin{tikzpicture}[baseline=11.5ex]
                \draw [->,thick,dotted] (-4mm,0) -- (-4mm,1mm) to[out=90,in=100] (8mm,-1mm) to[out=-80,in=-20,looseness=1.4] (-6mm,-13mm);
                \end{tikzpicture}
              }}}
        }
      }
      &&&&
      \infer=
      {\fun{f}^\omega \ieqdown \fun{f}(\fun{b})} 
      {
        \infer=
        {\fun{f}^\omega \ieq \fun{b}}
        {
          \infer=
          {\fun{f}^\omega \ieqdown \fun{f}(\fun{b})}
          {{\makebox(0,0){
                \hspace{-13mm}\begin{tikzpicture}[baseline=11.5ex]
                \draw [->,thick,dotted] (4mm,0) -- (4mm,1mm) to[out=90,in=80] (-8mm,-1mm) to[out=-100,in=-160,looseness=1.6] (6mm,-13mm);
                \end{tikzpicture}
              }}
          }
          &
          \fun{f}(\fun{b}) \leftarrow_{\varepsilon} \fun{b}
        }
      }
      &
      \fun{f}(\fun{b}) \leftarrow_{\varepsilon} \fun{b}
    }
    \\[2ex] %
    \infer=
    {\fun{C}(\fun{a}) \ieq \fun{C}^\omega}
    {
      \infer=
      {\fun{C}(\fun{a}) \ieqdown \fun{C}(\fun{b})} 
      {\infer={\fun{a} \ieq \fun{b}}{\text{(as above)}}}
      &
      \fun{C}(\fun{b}) \rstep \fun{C}(\fun{C}(\fun{a}))
      &
      \infer=
      {\fun{C}(\fun{C}(\fun{a})) \ieqdown \fun{C}^\omega}
      {\infer={\fun{C}(\fun{a}) \ieq \fun{C}^\omega}
            {\makebox(0,0){
              \hspace{18mm}\begin{tikzpicture}[baseline=13ex]
              \draw [->,thick,dotted] (0,0) -- (0,1mm) to[out=90,in=100] (15mm,-1mm) to[out=-80,in=-25,looseness=1.4] (-0mm,-14mm);
              \end{tikzpicture}
            }}
      }
    }
    \end{gather*}\vspace{-5ex}
    \caption{An example of infinitary equational reasoning, 
      deriving $\fun{C}(\fun{a}) \ieq \fun{C}^\omega$ in the TRS~$\aes$ of Example~\ref{ex:ieq}.
      Recall Notation~\ref{not:transitivity}.}
    \label{fig:ieq:fagb}
  \end{figure}
\end{example}

Definition~\ref{def:ieq:rules} of \,$\ieq$\, can also be defined using 
a greatest fixed point as follows: 
\begin{align*}
  {\ieq} \;\;\defd\;\; \gfp{R}{(\leftarrow_{\varepsilon} \cup \rstep \cup \mathrel{\down{R}})^*}  \,,
\end{align*}
where $\down{R}$ was defined in Definition~\ref{def:lifting}.
The equivalence of these definitions 
can be established in a similar way as in Theorem~\ref{thm:equiv:der:fp}.
It is easy to verify that the function 
$R \mapsto (\leftarrow_{\varepsilon} \cup \rstep \cup \mathrel{\down{R}})^*$
is monotone, and consequently the greatest fixed point exists. 

We note that, in the presence of collapsing rules (i.e., rules $\ell \to r$ where $r \in \avars$),
everything becomes equivalent: $s \ieq t$ for all terms $s,t$. 
For example, having a rule $\tail(x) \to x$ we obtain that
$s \ieq \tail(s) \ieq \tail^2(s) \ieq \cdots \ieq \tail^\omega$ for every term $s$.
This can be overcome by forbidding certain infinite terms and certain infinite limits.

\subsection{Bi-Infinite Rewriting}%

Another notion that arises naturally in our setup
is that of bi-infinite rewriting,
allowing rewrite sequences to extend infinitely forwards and backwards.
We emphasize that each of the steps $\rstep$ 
in such sequences is a forward step.

\begin{definition}\label{def:ibi:rules}
  Let $\atrs$ be a term rewriting system over $\Sigma$,
  and let $T= \iter{\Sigma}{\avars}$.
  We define the \emph{bi-infinite rewrite relation ${\ibi} \subseteq T \times T$} 
  by the following coinductive rules
  \begin{align*}
    \infer=
    {s \ibi t}
    {s \mathrel{(\rstep \cup \ibidown)^*} t}
    &&  
    \infer=
    {f(s_1,s_2,\ldots,s_n) \ibidown f(t_1,t_2,\ldots,t_n)}
    {s_1 \ibi t_1 & \cdots & s_n \ibi t_n}
  \end{align*}
  where ${\ibidown} \subseteq T \times T$ stands for bi-infinite rewriting below the root.
\end{definition}

If we replace $\ieq$ and $\ired$ by $\ibi$, and $\ieqdown$ and $\ireddown$ by $\ibidown$,
then Examples~\ref{ex:Ca:a} and~\ref{ex:fab} %
are illustrations of this rewrite relation.

Again, like $\ieq$, the relation $\ibi$\, can also be defined using 
a greatest fixed point: %
\begin{align*}
  {\ibi} &\;\;\defd\;\; \gfp{R}{(\rstep \cup \mathrel{\down{R}})^*} \,.
\end{align*}
Monotonicity of
$R \mapsto (\rstep \cup \mathrel{\down{R}})^*$
is easily verified, so that the greatest fixed point exists. 
Also, the equivalence of Definition~\ref{def:ibi:rules} with this $\sgfp$-definition 
can be established similarly. %

\section{Relating the Notions}\label{sec:venn}

\begin{lemma}
  Each of the relations $\ired$, $\ibi$ and $\ieq$ is reflexive and transitive.
  The relation $\ieq$ is also symmetric.
\end{lemma}
\begin{proof}
  Follows immediately from
  the fact that the relations are defined using the reflexive-transitive closure in each of their first rules. 
\end{proof}

\newcommand{\ibii}{\mathrel{\reflectbox{$\ibi$}}}
\begin{theorem}
  For every TRS $\atrs$ we have the following inclusions:
  \begin{center}
    \vspace{-1ex}
    \begin{tikzpicture}[node distance=20mm]
      \node (ired) {$\ired$};
      \node (ibi) [right of=ired] {$\ibi$};
      \node (iredconv) [right of=ired,yshift=-6mm] {$({\iredi} \cup {\ired})^*$};
      \node (ibiconv) [right of=ibi,node distance=25mm] {$({\ibii} \cup {\ibi})^*$};
      \node (ieq) [right of=ibiconv,node distance=25mm] {$\ieq$};
      
      \node at ($(ired)!.5!(ibi)$) {$\subseteq$};
      \node at ($(ibi)!.4!(ibiconv)$) {$\subseteq$};
      \node at ($(ired)!.6!(iredconv.west)$) [rotate=-30] {$\subseteq$};
      \node at ($(iredconv.east)!.4!(ibiconv.west)$) [rotate=20] {$\subseteq$};
      \node at ($(ibiconv)!.6!(ieq)$) {$\subseteq$};
    \end{tikzpicture}
    \vspace{-1ex}
  \end{center}
  Moreover, for each of these inclusions there exists a TRS for which the inclusion is strict. 
\end{theorem}

\begin{proof}
  The inclusions ${\ired} \subsetneq {\ibi} \subsetneq {\ieq}$ have already been established in the introduction.
  The inclusion ${\ired} \subsetneq {({\iredi} \cup {\ired})^*}$ is well-known (and obvious).
  Also ${\ibi} \subsetneq {({\ibii} \cup {\ibi})^*}$ is immediate since $\ibi$ is not symmetric.
  
  The inclusion ${({\iredi} \cup {\ired})^*} \subseteq {({\ibii} \cup {\ibi})^*}$
  is immediate since ${\ired} \subseteq {\ibi}$.
  Example~\ref{ex:Ca:a} witnesses strictness of this inclusion.
  The reason is that, for this example, ${\ired} = {\to^*}$ as the system does not admit any forward limits.
  Hence ${({\iredi} \cup {\ired})^*}$ is just finite conversion on potentially infinite terms.
  Thus $\fun{C}^\omega \ibi \fun{a}$, but not $\fun{C}^\omega \mathrel{({\iredi} \cup {\ired})^*} \fun{a}$.

  The inclusion ${({\ibii} \cup {\ibi})^*} \subseteq {\ieq}$
  follows from the fact that $\ieq$ includes $\ibi$ and is symmetric and transitive.
  Example~\ref{ex:ieq} witnesses strictness:
  $\fun{C}(\fun{a}) = \fun{C}^\omega$ can only be derived by
  a rewrite sequence of the form:
  \begin{align*}
    \fun{C}(\fun{a}) \ibi \fun{C}(\fun{f}^\omega) \ibiinv \fun{C}(\fun{b}) \to \fun{C}(\fun{C}(\fun{a}))
    \ibi \fun{C}(\fun{C}(\fun{f}^\omega)) \ibiinv \fun{C}(\fun{C}(\fun{b})) \to \fun{C}(\fun{C}(\fun{C}(\fun{a}))) \ibi \cdots
  \end{align*}
  and hence we need to change rewriting directions infinitely often  %
  whereas ${({\ibii} \cup {\ibi})^*}$ allows to change the direction only a finite number of times.
\end{proof}

\newcommand{\ieqclose}[1]{\stackrel{\infty}{T}(#1)}
\begin{lemma}\label{lem:ieq:closed}
  For relations $S \subseteq \iter{\Sigma}{\avars} \times \iter{\Sigma}{\avars}$ we define 
  \begin{align*}
    \ieqclose{S} \;\;\defd\;\; \gfp{R}{(S^{-1} \cup S \cup {\mathrel{\down{R}}})^*}\;.
  \end{align*}
  We have ${\ieqclose{S}} = {\ieqclose{\ieqclose{S}}}$ for every $S \subseteq \iter{\Sigma}{\avars} \times \iter{\Sigma}{\avars}$.
\end{lemma}

\begin{proof}
  For every relation $S$ we have
  $S \subseteq (S^{-1} \cup S \cup {\mathrel{\down{R}}})^*$
  and hence $S \subseteq {\ieqclose{S}}$ 
  by \eqref{eq:coind-ind-rules}.
  Hence it follows that ${\ieqclose{S}} \subseteq {\ieqclose{\ieqclose{S}}}$.
  For ${\ieqclose{\ieqclose{S}}} \subseteq {\ieqclose{S}}$ we note that
  \begin{align*}
    {\ieqclose{\ieqclose{S}}} 
    &= {(\;{\ieqclose{S}^{-1}} \cup {\ieqclose{S}} \cup {\down{\ieqclose{\ieqclose{S}}}}\;)^*} && \text{by definition}\\
    &= {(\;{\ieqclose{S}} \cup {\down{\ieqclose{\ieqclose{S}}}}\;)^*} && \text{by symmetry of $\ieqclose{S}$} \\
    &= {(\;(\;{S^{-1}} \cup S \cup {\down{\ieqclose{S}}}\;)^* \cup {\down{\ieqclose{\ieqclose{S}}}}\;)^*} && \text{by definition} \\
    &= {(\;{S^{-1}} \cup S \cup {\down{\ieqclose{S}}} \cup {\down{\ieqclose{\ieqclose{S}}}}\;)^*}  \\
    &= {(\;{S^{-1}} \cup S \cup {\down{\ieqclose{\ieqclose{S}}}}\;)^*} && \text{since ${\ieqclose{S}} \subseteq {\ieqclose{\ieqclose{S}}}$} 
  \end{align*}
  Thus ${\ieqclose{\ieqclose{S}}}$ is a fixed point of $R \mapsto (S^{-1} \cup S \cup {\mathrel{\down{R}}})^*$,
  and hence ${\ieqclose{\ieqclose{S}}} \subseteq {\ieqclose{S}}$.
\end{proof}

It follows immediately that $\ieq$ is closed under $\ieqclose{\cdot}$.
\begin{corollary}
  We have ${\ieq} \;=\; {\ieqclose{\ieq}}$ for every TRS $\atrs$.
\end{corollary}

\begin{proof}
  We have 
  ${\ieq} \;=\; {\ieqclose{\to_{\varepsilon}}} 
          \;=\; {\ieqclose{\ieqclose{\to_{\varepsilon}}}} 
          \;=\; {\ieqclose{\ieq}}$.
\end{proof}

The work~\cite{kahr:2013} introduces various notions of infinitary rewriting.
We comment on the notions that are closest to the relations $\ibi$ and $\ieq$ introduced in our paper.
First, we note that it is not difficult to see that ${\ibi} \subsetneq {\mred_t}$ where
$\mred_t$ is the topological graph closure of $\to$.
The paper~\cite{kahr:2013} also introduces a notion of infinitary equational reasoning
with a strongly convergent flavour, namely:
\begin{align*}
  S_E(R) = (S \circ E)^*(R)
\end{align*}
where 
$E(R)$ is the equivalence closure of $R$, and
$S(R)$ is the strongly convergent rewrite relation obtained from (single steps) $R$. 
The following lemma relates this notion with our notion of infinitary equational reasoning $\ieq$.

\begin{lemma}
  We have ${S_E(\to)} \;\subseteq\; {\ieq}$ for every TRS $\atrs$. 
  Moreover, there exists TRSs $\atrs$ for which the inclusion is strict.
\end{lemma}

\begin{proof}
  The inclusion is immediate from 
  \begin{enumerate}
    \item ${\to} \subseteq {\ieq}$,
    \item ${E(\ieq)} = {\ieq}$, and
    \item ${S(\ieq)} \subseteq {\ieqclose{\ieq}} = {\ieq}$.
  \end{enumerate}
  The following example shows that the inclusion can be strict.
\end{proof}

\begin{example}
  Consider the TRS $\atrs$ consisting of the rules
  \begin{align*}
    g(b_i(\varepsilon)) &\to g(g(a_{i+1}(\varepsilon)))\\
    a_0(\varepsilon) &\to f_0(a_0(\varepsilon)) &
    b_0(\varepsilon) &\to f_0(b_0(\varepsilon)) \\
    a_{i+1}(\varepsilon) &\to a_{i+1}(a_i(\varepsilon)) &
    b_{i+1}(\varepsilon) &\to b_{i+1}(b_i(\varepsilon)) \\
    a_{i+1}(b_i(\varepsilon)) &\to f_{i+1}(a_{i+1}(\varepsilon)) &
    b_{i+1}(a_i(\varepsilon)) &\to f_{i+1}(b_{i+1}(\varepsilon)) 
  \end{align*}
  for every $i \in \nat$.
  We will argue that
  \begin{align}
    g(a_0(\varepsilon)) \ieq g^\omega \;, \label{eq:g0a0}
  \end{align}
  but this equality does not hold in $S_E(\to)$.
  
  We have $a_0(\varepsilon) \ieq b_0(\varepsilon)$ as a consequence of $a_0(\varepsilon) \ieq f_0^\omega \ieq b_0(\varepsilon)$. 
  Moreover, by induction on $i \in \nat$ we conclude that $a_{i+1}(\varepsilon) \ieq b_{i+1}(\varepsilon)$  
  since
  \begin{align*}
    a_{i+1}(\varepsilon) \ieq a_{i+1}(a_i(\varepsilon)) \ieq a_{i+1}(b_i(\varepsilon)) \ieq f_{i+1}(a_{i+1}(\varepsilon)) \ieq \cdots \ieq f_{i+1}^\omega 
  \end{align*}
  and likewise
  \begin{align*}
    b_{i+1}(\varepsilon) \ieq b_{i+1}(b_i(\varepsilon)) \ieq b_{i+1}(a_i(\varepsilon)) \ieq f_{i+1}(b_{i+1}(\varepsilon)) \ieq \cdots \ieq f_{i+1}^\omega \;.
  \end{align*}
  Then we obtain equation~\eqref{eq:g0a0} since 
  \begin{align*}
    g(a_0(\varepsilon)) \ieq g(b_0(\varepsilon)) \ieq g(g(a_1(\varepsilon))) \ieq g(g(b_1(\varepsilon))) \ieq g(g(g(a_2(\varepsilon)))) \ieq \cdots \ieq g^\omega
  \end{align*}
  We give a rough sketch of the proof that equation~\eqref{eq:g0a0} is not valid in $S_E(\to) = (S \circ E)^*(\to)$.
  We use $R_n$ to denote the relation $(S \circ E)^n(\to)$. 
  Note that 
  \begin{enumerate}[label=(\alph*)]
    \item 
      The relation $R_1$ contains only forward limits such as 
      $a_0(\varepsilon) \mathrel{R_1} f_0^\omega$ and
      $b_0(\varepsilon) \mathrel{R_1} f_0^\omega$,
      but no backwards limits like
      $f_0^\omega \mathrel{R_1} b_0(\varepsilon)$ and hence not $a_0(\varepsilon) \mathrel{R_1} b_0(\varepsilon)$.
    \item 
      The relation $R_2$ contains forward limits using steps $E(R_1)$ 
      and we have $a_0(\varepsilon) \mathrel{E(R_1)} b_0(\varepsilon)$.
      Thus we can derive
      $a_1(\varepsilon) \mathrel{R_2} f_1^\omega$ and
      $b_1(\varepsilon) \mathrel{R_2} f_1^\omega$.
      However, we do not have 
      $f_1^\omega \mathrel{R_2} b_1(\varepsilon)$.
    \item 
      The relation $R_3$ contains forward limits using steps $E(R_2)$ 
      and we have $a_1(\varepsilon) \mathrel{E(R_2)} b_1(\varepsilon)$.
      We derive 
      $a_2(\varepsilon) \mathrel{R_3} f_2^\omega$ and
      $b_2(\varepsilon) \mathrel{R_3} f_2^\omega$, but not
      $f_2^\omega \mathrel{R_3} b_2(\varepsilon)$.
    \item \ldots
  \end{enumerate}
  
  In general, we obtain for every $i,j \in \nat$ that
  \begin{enumerate}
    \item $a_j(\varepsilon) \mathrel{R_i} f_j^\omega$ if and only if $j < i$, 
    \item $b_j(\varepsilon) \mathrel{R_i} f_j^\omega$ if and only if $j < i$, 
    \item $a_j(\varepsilon) \mathrel{E(R_i)} b_j(\varepsilon)$ if and only if $j < i$.
  \end{enumerate} 

  For deriving equation~\eqref{eq:g0a0} we need $a_i(\varepsilon) = b_i(\varepsilon)$ for every $i \in \nat$,
  and consequently there exists no $i \in \nat$ such that $g(a_0(\varepsilon)) \mathrel{R_i} g^\omega$.
  Hence equation~\eqref{eq:g0a0} does not hold in $S_E(\to)$.

  We note that the rewrite system contains an infinite number of rules and symbols.
  However, this is not crucial for illustrating the difference between $\ieq$ and $S_E(\to)$. 
  We can obtain a system with finitely many rules and symbols
  by modelling $a_i(x)$, $b_i(x)$ and $f_i(x)$ by $a(s^i(0),x)$, $b(s^i(0),x)$ and $f(s^i(0),x)$.
  The corresponding finite system is:
  \begin{align*}
    g(b(x,\varepsilon)) &\to g(g(a(s(x),\varepsilon)))\\
    a(0,\varepsilon) &\to f(0,a(0,\varepsilon)) &
    b(0,\varepsilon) &\to f(0,b(0,\varepsilon)) \\
    a(s(x),\varepsilon) &\to a(s(x),a(x,\varepsilon)) &
    b(s(x),\varepsilon) &\to b(s(x),b(x,\varepsilon)) \\
    a(s(x),b(x,\varepsilon)) &\to f(s(x),a(s(x),\varepsilon)) &
    b(s(x),a(x,\varepsilon)) &\to f(s(x),b(s(x),\varepsilon)) 
  \end{align*}
  Then $g(a(0,\varepsilon)) \ieq g^\omega$ but this equation does not hold in $S_E(\to)$.
\end{example}

\section{A Formalization in Coq}\label{sec:coq}

The standard definition of infinitary rewriting,
using ordinal length rewrite sequences and strong convergence at limit ordinals,
is difficult to formalize. 
The coinductive framework we propose, is easy to formalize and work with in theorem provers.

In Coq, the coinductive definition of 
infinitary strongly convergent reductions can be defined as follows:
{\small
\begin{verbatim}
Inductive ired : relation term :=
  | Ired : 
      forall R I : relation term,
      subrel I ired ->
      subrel R ((root_step (+) lift I)* ;; lift R) ->
      subrel R ired.
\end{verbatim}}
\noindent
Here 
\verb=term= is the set of coinductively defined terms,
\verb=;;= is relation composition,
\verb=(+)= is the union of relations, 
\verb=*= the reflexive-transitive closure,
\verb=lift R= is~$\down{R}$,
and \verb=root_step= is the root step relation.

Let us briefly comment on this formalization.
Recall that  ${\ired} \;\defd\; \lfp{R}{\gfp{S}{G(R,S)}}$
where $G(R,S) = (\rstep \cup \mathrel{\down{R}})^*\relcomp \down{S}$.
The inductive definition of \verb=ired= corresponds to the least fixed point $\slfp{R}$.
Coq has no support for mutual inductive and coinductive definitions.
Therefore, instead of the explicit coinduction, 
we use the $\nu$-rule from~\eqref{eq:coind-ind-rules}.
For every relation $T$
that fulfills $T \subseteq  G(R,T)$, we have that $T \subseteq \gfp{S}{G(R,S)}$.
Moreover, we know that $\gfp{S}{G(R,S)}$ is the union of all these relations $T$.
Finally, we introduce an auxiliary relation \verb=I= 
to help Coq generate a good induction principle.
One can think of \verb=I= as consisting of those pairs
for which the recursive call to \verb=ired= is invoked.
Replacing \verb=lift I= by \verb=lift ired= is correct, 
but then the induction principle that Coq generates for \verb=ired= is useless.

On the basis of the above definition we proved the Compression Lemma:
whenever there is an infinite reduction from $s$ to $t$ ($s \ired t$) 
then there exists a reduction of length at most $\omega$ from $s$ to~$t$ ($s \to^{\le \omega} t$).
The Compression Lemma holds for left-linear TRSs with finite left-hand sides.
To characterize rewrite sequences $\to^{\le \omega}$ in Coq, we define:
{\small
\begin{verbatim}
  Inductive ored : relation (term F X) :=
  | Ored : 
      forall R : relation (term F X),
      subrel R (mred ;; lift R) ->
      forall s t, R s t -> ored s t.
\end{verbatim}}
\noindent
Here \verb=mred= are finite rewrite sequences $\to^*$.
The definition can be understood as follows.
We want the relation~\verb=ored= to be the greatest fixed point of $H$ defined by
$H(R) =  {\to^*} \circ {\down{R}}$.
So we allow a finite rewrite sequence after which the rewrite activity has to go `down' to the arguments.
Again, as above for \verb=ired=, we avoid the use of coinduction
and define \verb=ored= inductively as the union of all relations $R$ with $R \subseteq H(R)$.

To the best of our knowledge this is the first formal proof of this well-known lemma.
The formalization is available at 
\url{http://dimitrihendriks.com/coq/compression}.

\section{Conclusion}\label{sec:conclusion}
We have proposed a coinductive framework which gives rise to several
natural variants of infinitary rewriting in a uniform way:
  \begin{enumerate}[label=({\alph*})]
    \item infinitary equational reasoning
        ${\ieq} \;\;\defd\;\; \gfp{y}{(\leftarrow_{\varepsilon} \cup \rstep \cup \mathrel{\down{y}})^*}$,
  \smallskip
    \item bi-infinite rewriting
        ${\ibi} \;\;\defd\;\; \gfp{y}{(\rstep \cup \mathrel{\down{y}})^*}$, and
  \smallskip
    \item infinitary rewriting
        ${\ired} \;\;\defd\;\; \lfp{x}{\gfp{y}{(\rstep \cup \mathrel{\down{x}})^*\relcomp \down{y}}}$\,.
  \end{enumerate}
We believe that (a) and (b) are new.
As a consequence of the coinduction over the term structure,
these notions have the strong convergence built-in,
and thus can profit from the well-developed techniques (such as tracing)
in infinitary rewriting.

We have given a mixed inductive/coinductive definition of infinitary rewriting 
and established a bridge between infinitary rewriting and coalgebra.
Both fields are concerned with infinite objects and we would like %
to understand their relation better. 
In contrast to previous coinductive treatments, 
the framework presented here captures rewrite sequences of arbitrary ordinal length,
and paves the way for formalizing infinitary rewriting in theorem provers
(as illustrated by our proof of the Compression Lemma in Coq).

Concerning proof trees/terms for infinite reductions, let us mention that 
an alternative approach has been developed in parallel by Lombardi, R\'{\i}os and de~Vrijer~\cite{lomb:rios:vrij:2014}.
While we focus on proof terms for the reduction relation and abstract from the order of steps in parallel subterms, 
they use proof terms for modeling the fine-structure of 
the infinite reductions themselves.
Another difference is that our framework allows for non-left-linear systems.
We believe that both approaches are complementary.
Theorems for which the fine-structure of rewrite sequences is crucial,
must be handled using~\cite{lomb:rios:vrij:2014}.
(But note that we can capture standard reductions by a restriction on proof trees 
and prove standardization using proof tree transformations, see~\cite{endr:polo:2012b}). 
If the fine-structure is not important, as for instance for proving confluence,
then our system is more convenient to work with due to simpler proof terms.

Our work lays the foundation for several directions of future research:
\begin{enumerate}
  \item 
    The coinductive treatment of infinitary $\lambda$-calculus~\cite{endr:polo:2012b}
    has led to elegant, significantly simpler proofs~\cite{czaj:2014,czaj:2015}
    of some central properties of the infinitary $\lambda$-calculus.
    The coinductive framework that we propose
    enables similar developments for infinitary term rewriting
    with reductions of arbitrary ordinal length.

  \item 
    The concepts of bi-infinite rewriting and infinitary equational reasoning are novel.
    We would like to study these concepts,
    in particular since the theory of infinitary equational reasoning is still underdeveloped.
    For example, it would be interesting to compare the 
    Church--Rosser properties 
    \begin{align*}
      {\ieq}  \;\subseteq\; {\ired \relcomp \iredi} 
      &&\text{ and }&& 
      {(\iredi \relcomp \ired)^*} \;\subseteq\; {\ired \relcomp \iredi} \;\,.
    \end{align*}
    \vspace{-4ex}
    
  \item 
    The formalization of the proof of the Compression Lemma in Coq
    is just the first step towards the formalization of all major theorems in infinitary rewriting.
    
  \item It is interesting to investigate whether and how the coinductive framework
    can be extended to other notions of infinitary rewriting,
    for example reductions %
    where root-active terms are mapped to $\bot$ in the limit~\cite{bahr:2010,bahr:2010b,bahr:2012,endr:hend:klop:2012}.
    
  \item 
    We believe that the coinductive definitions will ease the development of new techniques for automated reasoning 
    about infinitary rewriting.
    For example, methods
    for proving (local) productivity~\cite{endr:grab:hend:2009,endr:hend:2011,zant:raff:2010},
    for (local) infinitary normalization~\cite{zant:2008,endr:grab:hend:klop:vrij:2009,endr:vrij:wald:10},
    for (local) unique normal forms~\cite{endr:hend:grab:klop:oost:2014},
    and for analysis of infinitary reachability and infinitary confluence.
    Due to the coinductive definitions, the implementation and formalization of these techniques 
    could make use of circular coinduction~\cite{gogu:lin:rosu:2000,endr:hend:bodi:2013}.  
\end{enumerate}

\subsection*{Acknowledgments}
We thank Patrick Bahr, Jeroen Ketema, and Vincent van Oostrom for fruitful discussions 
and comments on earlier versions of this paper.

\bibliography{main}

\begin{thebibliography}{10}

\bibitem{baad:nipk:1998}
F.~Baader and T.~Nipkow.
\newblock {\em {Term Rewriting and All That}}.
\newblock Cambridge Univ. Press, 1998.

\bibitem{bahr:2010b}
P.~Bahr.
\newblock {Abstract Models of Transfinite Reductions}.
\newblock In {\em Proc.\ Conf.\ on Rewriting Techniques and Applications (RTA
  2010)}, volume~6 of {\em Leibniz International Proceedings in Informatics},
  pages 49--66. Schloss Dagstuhl, 2010.

\bibitem{bahr:2010}
P.~Bahr.
\newblock {Partial Order Infinitary Term Rewriting and B{\"o}hm Trees}.
\newblock In {\em Proc.\ Conf.\ on Rewriting Techniques and Applications (RTA
  2010)}, volume~6 of {\em Leibniz International Proceedings in Informatics},
  pages 67--84. Schloss Dagstuhl, 2010.

\bibitem{bahr:2012}
P.~Bahr.
\newblock {Infinitary Term Graph Rewriting is Simple, Sound and Complete}.
\newblock In {\em Proc.\ Conf.\ on Rewriting Techniques and Applications (RTA
  2012)}, volume~15 of {\em Leibniz International Proceedings in Informatics},
  pages 69--84. Schloss Dagstuhl, 2012.

\bibitem{bare:1977}
H.P. Barendregt.
\newblock {The Type Free Lambda Calculus}.
\newblock In {\em Handbook of Mathematical Logic}, pages 1091--1132.
  Nort-Holland Publishing Company, Amsterdam, 1977.

\bibitem{bare:klop:2009}
H.P. Barendregt and J.W. Klop.
\newblock {Applications of Infinitary Lambda Calculus}.
\newblock {\em Information and Computation}, 207(5):559--582, 2009.

\bibitem{coqu:1996}
C.~Coquand and Th. Coquand.
\newblock {On the Definition of Reduction for Infinite Terms}.
\newblock {\em Comptes Rendus de l'Acad\'emie des Sciences. S\'erie I},
  323(5):553--558, 1996.

\bibitem{coqu:1994}
Th. Coquand.
\newblock Infinite objects in type theory.
\newblock In Henk Barendregt and Tobias Nipkow, editors, {\em Types for Proofs
  and Programs, International Workshop TYPES'93, Nijmegen, The Netherlands, May
  24--28, 1993, Selected Papers}, volume 806 of {\em LNCS}, pages 62--78.
  Springer, 1994.

\bibitem{czaj:2014}
\L. Czajka.
\newblock {A Coinductive Confluence Proof for Infinitary Lambda-Calculus}.
\newblock In {\em Rewriting and Typed Lambda Calculi (RTA-TLCA 2014)}, volume
  8560 of {\em Lecture Notes in Computer Science}, pages 164--178. Springer,
  2014.

\bibitem{czaj:2015}
{\L}.~{Czajka}.
\newblock {Coinductive Techniques in Infinitary Lambda-Calculus}.
\newblock {\em ArXiv e-prints}, 2015.

\bibitem{ders:kapl:plai:1991}
N.~Dershowitz, S.~Kaplan, and D.A. Plaisted.
\newblock {Rewrite, Rewrite, Rewrite, Rewrite, Rewrite,\dots}.
\newblock {\em Theoretical Computer Science}, 83(1):71--96, 1991.

\bibitem{endr:vrij:wald:10}
J.~Endrullis, R.~C. de~Vrijer, and J.~Waldmann.
\newblock {Local Termination: Theory and Practice}.
\newblock {\em Logical Methods in Computer Science}, 6(3), 2010.

\bibitem{endr:grab:hend:2009}
J.~Endrullis, C.~Grabmayer, and D.~Hendriks.
\newblock {Complexity of Fractran and Productivity}.
\newblock In {\em Proc.\ Conf.\ on Automated Deduction (CADE~22)}, volume 5663
  of {\em LNCS}, pages 371--387, 2009.

\bibitem{endr:grab:hend:klop:vrij:2009}
J.~Endrullis, C.~Grabmayer, D.~Hendriks, J.W. Klop, and R.C de~Vrijer.
\newblock {Proving Infinitary Normalization}.
\newblock In {\em Postproc.\ Int.\ Workshop on Types for Proofs and Programs
  (TYPES 2008)}, volume 5497 of {\em LNCS}, pages 64--82. Springer, 2009.

\bibitem{endr:hans:hend:polo:silv:2013}
J.~Endrullis, H.~Hvid Hansen, D.~Hendriks, A.~Polonsky, and A.~Silva.
\newblock {A Coinductive Treatment of Infinitary Rewriting}.
\newblock {\em CoRR}, abs/1306.6224, 2013.

\bibitem{endr:hend:2011}
J.~Endrullis and D.~Hendriks.
\newblock {Lazy Productivity via Termination}.
\newblock {\em Theoretical Computer Science}, 412(28):3203--3225, 2011.

\bibitem{endr:hend:bodi:2013}
J.~Endrullis, D.~Hendriks, and M.~Bodin.
\newblock {Circular Coinduction in Coq Using Bisimulation-Up-To Techniques}.
\newblock In {\em Proc. Conf. on Interactive Theorem Proving (ITP)}, volume
  7998 of {\em LNCS}, pages 354--369. Springer, 2013.

\bibitem{endr:hend:grab:klop:oost:2014}
J.~Endrullis, D.~Hendriks, C.~Grabmayer, J.W. Klop, and V.~van Oostrom.
\newblock Infinitary term rewriting for weakly orthogonal systems: Properties
  and counterexamples.
\newblock {\em Logical Methods in Computer Science}, 10(2:7):1--33, 2014.

\bibitem{endr:hend:klop:2012}
J.~Endrullis, D.~Hendriks, and J.W. Klop.
\newblock {Highlights in Infinitary Rewriting and Lambda Calculus}.
\newblock {\em Theoretical Computer Science}, 464:48--71, 2012.

\bibitem{endr:polo:2012b}
J.~Endrullis and A.~Polonsky.
\newblock {Infinitary Rewriting Coinductively}.
\newblock In {\em Proc.\ Types for Proofs and Programs (TYPES 2012)}, volume~19
  of {\em Leibniz International Proceedings in Informatics}, pages 16--27.
  Schloss Dagstuhl, 2013.

\bibitem{gogu:lin:rosu:2000}
J.~Goguen, K.~Lin, and G.~Ro\c{s}u.
\newblock {Circular Coinductive Rewriting}.
\newblock In {\em Proc. of Automated Software Engineering}, pages 123--131.
  IEEE, 2000.

\bibitem{jaco:rutt:2011}
B.~Jacobs and J.J.M.M. Rutten.
\newblock {An Introduction to (Co)Algebras and (Co)Induction}.
\newblock In {\em Advanced Topics in Bisimulation and Coinduction}, pages
  38--99. Cambridge University Press, 2011.

\bibitem{joac:2004}
F.~Joachimski.
\newblock {Confluence of the Coinductive Lambda Calculus}.
\newblock {\em Theoretical Computer Science}, 311(1-3):105--119, 2004.

\bibitem{kahr:2013}
S.~Kahrs.
\newblock {Infinitary Rewriting: Closure Operators, Equivalences and Models}.
\newblock {\em Acta Informatica}, 50(2):123--156, 2013.

\bibitem{kenn:vrie:2003}
J.R. Kennaway and F.-J. de~Vries.
\newblock {\em {Infinitary Rewriting}}, chapter~12.
\newblock Cambridge University Press, 2003.
\newblock in~\cite{tere:2003}.

\bibitem{kenn:klop:slee:vrie:1995a}
J.R. Kennaway, J.W. Klop, M.R. Sleep, and F.-J. de~Vries.
\newblock {Transfinite Reductions in Orthogonal Term Rewriting Systems}.
\newblock {\em Information and Computation}, 119(1):18--38, 1995.

\bibitem{kete:simo:2013}
J.~Ketema and J.G. Simonsen.
\newblock {Computing with Infinite Terms and Infinite Reductions}.
\newblock Unpublished manuscript.

\bibitem{klop:1992}
J.W. Klop.
\newblock {Term Rewriting Systems}.
\newblock In {\em Handbook of Logic in Computer Science}, volume~II, pages
  1--116. Oxford University Press, 1992.

\bibitem{klop:vrij:2005}
J.W. Klop and R.C de~Vrijer.
\newblock {Infinitary Normalization}.
\newblock In {\em We Will Show Them: Essays in Honour of Dov Gabbay (2)}, pages
  169--192. {College Publications}, 2005.

\bibitem{lomb:rios:vrij:2014}
C.~Lombardi, A.~R{\'{\i}}os, and R.C de~Vrijer.
\newblock {Proof Terms for Infinitary Rewriting}.
\newblock In {\em Rewriting and Typed Lambda Calculi (RTA-TLCA 2014)}, volume
  8560 of {\em Lecture Notes in Computer Science}, pages 303--318. Springer,
  2014.

\bibitem{simo:2004}
J.G. Simonsen.
\newblock {On Confluence and Residuals in Cauchy Convergent Transfinite
  Rewriting}.
\newblock {\em Information Processing Letters}, 91(3):141--146, 2004.

\bibitem{tere:2003}
Terese.
\newblock {\em {Term Rewriting Systems}}, volume~55 of {\em Cambridge Tracts in
  Theoretical Computer Science}.
\newblock Cambridge University Press, 2003.

\bibitem{verm:2010}
M.~Vermaat.
\newblock {Infinitary Rewriting in Coq}.
\newblock Available at url \url{http://martijn.vermaat.name/master-project/}.

\bibitem{zant:2008}
H.~Zantema.
\newblock {Normalization of Infinite Terms}.
\newblock In {\em Proc.\ Conf.\ on Rewriting Techniques and Applications (RTA
  2008)}, number 5117 in LNCS, pages 441--455, 2008.

\bibitem{zant:raff:2010}
H.~Zantema and M.~Raffelsieper.
\newblock {Proving Productivity in Infinite Data Structures}.
\newblock In {\em Proc.\ Conf.\ on Rewriting Techniques and Applications (RTA
  2010)}, volume~6 of {\em Leibniz International Proceedings in Informatics},
  pages 401--416. Schloss Dagstuhl, 2010.

\end{thebibliography}

\end{document}